\documentclass[
  reprint,
  superscriptaddress,
  amsmath,
  amssymb,
  aps,
  prx,
  nofootinbib
]{revtex4-2}

\usepackage[utf8]{inputenc}
\usepackage{mathtools}
\usepackage{braket}
\usepackage{verbatim}
\usepackage{xcolor}
\usepackage[normalem]{ulem}
\usepackage{enumitem}
\usepackage{tikz}
\usetikzlibrary{arrows.meta,positioning}
\usepackage{hyperref}
\usepackage{cleveref}
\usepackage[most]{tcolorbox}

\newtcolorbox{summarybox}{
  colback=gray!8,
  colframe=black,
  boxrule=0.6pt,
  arc=2pt,
  left=8pt,
  right=8pt,
  top=6pt,
  bottom=6pt,
  before skip=8pt,
  after skip=8pt
}

\crefname{problem}{Problem}{Problems}
\Crefname{problem}{Problem}{Problems}

\DeclareSymbolFont{bbold}{U}{bbold}{m}{n}
\DeclareMathAlphabet{\mathbbold}{U}{bbold}{m}{n}

\newcommand{\ind}{\mathbbold{1}}

\usepackage{amsthm}

\newtheorem{theorem*}{Theorem}
\newtheorem{theorem}{Theorem}
\newtheorem{lemma}{Lemma}
\newtheorem{proposition}{Proposition}
\newtheorem{corollary}{Corollary}

\theoremstyle{definition}
\newtheorem{definition}{Definition}
\newtheorem{example}{Example}
\newtheorem{problem}{Problem}
\theoremstyle{remark}
\newtheorem{remark}{Remark}

\newcommand{\PSPACE}{\mathsf{PSPACE}}

\newcommand{\SigmaTwoP}{\ensuremath{\Sigma_2^p}}
\renewcommand{\P}{\mathsf{P}}
\newcommand{\MA}{\mathsf{MA}}

\newcommand{\QMA}{\mathsf{QMA}}
\newcommand{\StoqMA}{\mathsf{StoqMA}}
\newcommand{\VGPMA}{\mathsf{VGPMA}}
\newcommand{\VGPLH}{\mathsf{VGP\textnormal{-}LH}}
\newcommand{\STOQLH}{\mathsf{Stoq\textnormal{-}LH}}
\newcommand{\NP}{\mathsf{NP}}
\newcommand{\coNP}{\mathsf{coNP}}
\newcommand{\poly}{\mathrm{poly}}
\newcommand{\Tr}{\mathrm{Tr}}

\newcommand{\Z}{\mathbb{Z}}
\newcommand{\R}{\mathbb{R}}
\newcommand{\C}{\mathbb{C}}

\newcommand{\supp}{\mathrm{supp}}
\newcommand{\stoq}{\mathrm{stoq}}
\newcommand{\gadget}{\mathrm{gadget}}

\newcommand{\diag}{\mathrm{diag}}

\newcommand{\FREENCLREV}{\mathsf{FreeNCLRev}}
\newcommand{\DSTOQ}{\mathsf{LOCAL}\textnormal{-}\mathcal{D}\textnormal{-}\mathsf{STOQ}}
\newcommand{\DSTOQl}{\DSTOQ\text{[}\ell\text{]}}
\newcommand{\USTOQ}{\mathsf{LOCAL}\textnormal{-}\mathcal{U}\textnormal{-}\mathsf{STOQ}}
\newcommand{\AND}{{\textsc{And}}}
\newcommand{\OR}{{\textsc{Or}}}
\newcommand{\ANDOR}{{\AND/\OR}}

\newcommand{\PARTITION}{\mathsf{PARTITION}}

\begin{document}

\title{
    Dismantling the Stoquastic Dichotomy
}

\author{Armen Karakashian}
\email{armen@karakashian.com}
\affiliation{Department of Applied Mathematics \& Statistics, Stony Brook University, Stony Brook, New York 11794, USA}

\author{Itay Hen}
\email{itayhen@isi.edu}
\affiliation{Information Sciences Institute, University of Southern California, Marina del Rey, California 90292, USA}
\affiliation{Department of Physics and Astronomy and Center for Quantum Information Science \& Technology, University of Southern California, Los Angeles, California 90089, USA}
\affiliation{Department of Electrical and Computer Engineering, University of Southern California, Los Angeles, California 90089, USA}

\date{\today}

\begin{abstract}
    \noindent
    We challenge the notion that a stoquastic binary governs fundamental computational boundaries in quantum computing and classical simulation of quantum systems.
    We argue that vanishing geometric phase (VGP), a geometric condition on the Hamiltonian's transition graph, more adequately captures these boundaries.
    To distinguish VGP from stoquasticity, we construct VGP 3-local Hamiltonians that are formally hard to stoquastize, yet belong to a family admitting polynomial-time recognition of the VGP property.
    Without constructing a stoquastizing unitary, we prove that the local Hamiltonian problem is $\mathsf{StoqMA}$-complete under the promise that the input Hamiltonian has VGP, and that a frustration-free variant is in $\mathsf{MA}$ under the same promise.
    We use this result to argue that non-VGP is necessary for any claimed adiabatic advantage justified by escaping the $\mathsf{StoqMA}$ regime.
    Further, we identify natural settings where the VGP property can be recognized in polynomial time.
    In contrast, we show that recognition of VGP is $\mathsf{PSPACE}$-complete in general for geometrically local Hamiltonians.
    Our results show that the computational boundaries $\mathsf{MA} \subseteq \mathsf{StoqMA} \subseteq \mathsf{QMA}$ traditionally attributed to stoquasticity are better understood as boundaries between vanishing and non-vanishing geometric phase structure.
\end{abstract}

\maketitle

\section{Introduction}
\label{sec:intro}

The quantum sign problem is a central obstacle in computational
quantum physics~\cite{Troyer_2005,Loh_1990}: it is often introduced as a primary cause of inefficient classical quantum Monte Carlo (QMC) simulation of generic quantum many-body
systems.
However, the sign problem's pervasiveness extends beyond any particular numerical method. 
Fundamentally, the sign problem concerns whether a quantum-many body system can admit a description analogous to a classical probabilistic model, having consequences for quantum computational complexity, efficient classical simulation, and the power of adiabatic quantum algorithms.

A Hamiltonian is \emph{stoquastic} if all its off-diagonal matrix elements in the computational basis are nonpositive~\cite{Bravyi_2006}.
Stoquasticity of $H$ ensures that Boltzmann weights $e^{-\beta H}$ have nonnegative entries, making Monte Carlo sampling more tractable, and is the property that the complexity class $\StoqMA$\footnote{See~\Cref{def:stoqma}.} was canonically defined around~\cite{Bravyi_2006}.
The sign problem boundary among Hamiltonians is often framed as a dichotomy between stoquastic and non-stoquastic Hamiltonians, with stoquastic Hamiltonians being considered sign-problem-free. 
However, previous works \cite{Hen_2021,Babakhani_2025} have demonstrated that stoquasticity is a special case of sign-problem-freeness in a particular type of QMC algorithm (permutation matrix representation quantum Monte Carlo, or PMR-QMC).
In particular, in this context \emph{stoquasticity} depends on the choice of computational
basis and can be created or destroyed by a diagonal unitary transformation, whereas
sign-problem-freeness itself is invariant under such transformations and does not cross the
PMR-QMC sign problem boundary~\cite{Hen_2021}.
The class of non-stoquastic Hamiltonians proven to be sign-problem-free in this context are known as \emph{vanishing geometric phase} (VGP)  Hamiltonians (see~\Cref{def:vgp}).
This finding is the key motivation for exploring whether VGP defines the sign problem boundary in other contexts.

In the main result of this paper, we will find that VGP preserves the $\StoqMA$ level of the $\MA$ vs. $\StoqMA$ vs. $\QMA$\footnote{See~\Cref{def:ma} and ~\Cref{def:qma} for the standard definitions of the complexity classes.} hierarchy: the local Hamiltonian problem restricted to VGP Hamiltonians is $\StoqMA$-complete, demonstrating that these boundaries in Hamiltonian complexity are not fully captured by stoquasticity (\Cref{cor:stoqma_complete}, \Cref{cor:vgpma_collapse}).
More broadly, while stoquasticity is an algorithmically useful representation, it is not the intrinsic feature characterizing sign problem boundaries in PMR-QMC or quantum computational complexity.
What's more, we will find that VGP can also be computationally advantageous by explicitly constructing VGP local Hamiltonians that are provably hard to stoquastize via locality-preserving unitary transformations (\Cref{thm:vgp_hard_to_stoquastize}).
In other words, it may be favorable to \emph{not attempt} to find a stoquastizing unitary to simulate these Hamiltonians.
Lastly, in one of the final sections (\Cref{sec:aqc}), we will argue that non-VGP instantaneous Hamiltonians are \emph{necessary} for the standard route to universal adiabatic quantum computing\footnote{The standard route being a claimed adiabatic advantage that relies on leaving the $\StoqMA$ regime.}, while non-stoquastic ones are not even \emph{sufficient}.

\subsection{Sign-Problem-Freeness is Invariant Under Diagonal Unitary Transformation in PMR-QMC}
\label{subsec:stoq_wrong}

The first diagonal-unitary-invariant description of sign-problem-freeness was introduced in Ref.~\cite{Hen_2021}, motivated by understanding on which Hamiltonian inputs the Permutation Matrix Representation quantum Monte Carlo (PMR-QMC)~\cite{Gupta_2020} algorithm is efficient.
We will review the findings of Ref.~\cite{Hen_2021}.

Recall the PMR framework for quantum Monte Carlo~\cite{Gupta_2020}.
In PMR-QMC, the partition function and energy estimators are expressed as sums over closed walks in the transition graph $G_H$ (\Cref{def:transition_graph}), with each closed walk contributing a weight proportional to the product of off-diagonal matrix elements along its edges.
For these weights to be nonnegative (i.e., there is no destructive interference in the closed-walk expansion of the partition function), it is necessary and sufficient that the \emph{holonomy} of every closed cycle in $G_H$ be trivial.
The holonomy of a cycle $\gamma = (z_0, z_1, \ldots, z_q = z_0)$ is the complex number
\begin{align*}
  \Phi(\gamma) = (-1)^q \prod_{t=1}^{q}
  \frac{H_{z_t, z_{t-1}}}{|H_{z_t, z_{t-1}}|},
\end{align*}
encoding the total phase accumulated around the cycle.
A Hamiltonian $H$ has \emph{vanishing geometric phase} (VGP) if $\Phi(\gamma) = 1$ for every cycle in $G_H$; equivalently, if a diagonal unitary $D$ exists that renders $H$ stoquastic via conjugation (proven in Ref.~\cite{Hen_2021}).
As established in Ref.~\cite{Hen_2021}, a Hamiltonian is sign-problem-free under PMR-QMC if and only if it has VGP.
Stoquasticity is sufficient but not necessary: a Hamiltonian may be sign-problem-free without being stoquastic in any given basis (see \Cref{ex:twoqubit_vgp}, \Cref{ex:n_qubit_family_non_local_stoquastizing _diagonal}).
In particular, VGP is the diagonal-unitary-invariant notion of sign-problem-freeness in PMR-QMC, while stoquasticity is just one representative.
This result implies that the correct equivalence class for sign-problem-freeness in PMR-QMC is not characterized by stoquasticity, but rather by VGP---a deeper property of the geometric phase structure of the Hamiltonian's transition graph. 

Follow-up work~\cite{Babakhani_2025} explicitly identified cases where recognition of VGP in physically-motivated Hamiltonians is tractable, and also analyzed the scaling behavior of VGP-inspired sign problem diagnostics across the space of unitary transformations.
We intend to advance the results of these works by resolving open complexity-theoretic questions surrounding VGP: namely, is there a formal computational separation between stoquasticity and VGP, what is the complexity of deciding VGP, how does VGP fit into established boundaries in quantum computational complexity, and what does the VGP framework say about quantum advantage in the design of adiabatic quantum computers?

\subsection{Outline}
\label{subsec:this_work}

In~\Cref{sec:vgp}, we provide several different equivalent characterizations of VGP and we introduce the abstract \emph{stoquastic proxy map} $\mathcal{S}(\, \cdot \,)$ (\Cref{def:stoq_proxy_map}), which transforms VGP Hamiltonians into cospectral stoquastic Hamiltonians. We emphasize that $\mathcal{S}(\, \cdot \,)$ is not guaranteed to be efficiently computable in general---it is merely a helpful tool for the methods of this paper. We also review existing results on VGP as the exact QMC sign-problem-free boundary and provide explicit examples of VGP-but-not-stoquastic Hamiltonians.
For the main result of this section, we prove that there are VGP Hamiltonians that are provably hard to stoquastize via unitary transformation but admit efficient recognition of the VGP property.

In~\Cref{sec:pmr_locality}, we introduce a new notion of locality (PMR locality) that is motivated by the PMR formalism. PMR locality subsumes geometric locality, but is stricter than $k$-locality. We prove that PMR-local Hamiltonians (and, consequently, geometrically local Hamiltonians) promised to have VGP can be efficiently stoquastized. That is, one can efficiently obtain a compact representation of a cospectral stoquastic local Hamiltonian.

In \Cref{sec:vgpma}, we introduce the complexity class $\VGPMA$ (\Cref{def:vgpma}) defined as the class of problems reducible to the local Hamiltonian problem (\Cref{def:local_hamiltonian_problem}) under the promise that the input Hamiltonian has VGP.
We then present the main result of the paper, showing $\VGPMA = \StoqMA$ as a corollary of $\VGPLH$ being $\StoqMA$-complete.
The containment $\StoqMA \subseteq \VGPMA \subseteq \QMA$ is immediate because stoquasticity implies VGP.
The question of interest is where $\VGPMA$ exists in the complexity hierarchy with respect to $\StoqMA$ and $\QMA$.

A priori, there are three possible outcomes:
\begin{enumerate}[label=($\roman*$)]
\item[$(i)$] $\VGPMA = \StoqMA$
\item[$(ii)$] $\StoqMA \subsetneq \VGPMA \subsetneq \QMA$
\item[$(iii)$] $\VGPMA = \QMA$
\end{enumerate}
Each outcome is plausible: for outcomes $(ii)$ and $(iii)$, a VGP Hamiltonian $H$ may encode computational hardness by knowing that some diagonal unitary $D$ maps $H \mapsto H_{\stoq}$, but $D$ may not be efficiently realizable as a unitary circuit (see \Cref{ex:n_qubit_family_non_local_stoquastizing _diagonal}).
We define efficiently realizable unitaries in ~\Cref{def:eff_realizable}.
In this work, we show that the intermediate scenario outcome $(ii)$ is not possible. 
In particular, we prove that the local Hamiltonian problem is $\StoqMA$-complete for VGP Hamiltonians (\Cref{thm:vgp_stoqma}).
This immediately implies that outcome $(i)$ $\VGPMA = \StoqMA$ holds (\Cref{cor:vgpma_collapse}).
This result, together with the characterization from Ref.~\cite{Hen_2021}, shows that VGP is the appropriate notion of sign-problem-freeness for both PMR-QMC and Hamiltonian complexity.
We also prove that the frustration-free local Hamiltonian problem is in $\MA$ for PMR-local Hamiltonians that are promised to have VGP (\Cref{thm:ff-pmr-local-lh-ma}).

In~\Cref{sec:vgp_complexity}, we use the separate characterizations of VGP from \Cref{sec:vgp} to classify the computational complexity of recognizing VGP in local Hamiltonians. 
We prove that recognizing VGP in local Hamiltonians is $\PSPACE$-complete in general, but is $\coNP$-hard for 2-local.
$\coNP$-completeness in the 2-local case is conjectured (Problem \ref{prob:2_local_poly_diam}).
We complement these results in the next section by demonstrating when VGP recognition is tractable.

In~\Cref{sec:diagonal_recovery}, we examine when recovering a stoquastizing  diagonal can be done efficiently, allowing one to efficiently certify VGP under some natural conditions, particularly in certain 2-local, PMR-local and geometrically local settings (\Cref{thm:dstoq_efficient}, \Cref{cor:dstoq_ell}).

In~\Cref{sec:operational}, we identify computational tasks that separate VGP Hamiltonians from stoquastic Hamiltonians.

In~\Cref{sec:aqc}, we discuss the consequences of our results for the design of non-VGP adiabatic quantum computing algorithms.

In \Cref{sec:discussion}, we state conclusions, open problems, and future directions.

\section{Preliminaries: notation \& key concepts}
\label{sec:preliminaries}

We consider $n$-qubit systems with Hilbert space $\mathcal{H} = (\C^2)^{\otimes n}$ and computational basis $\{\ket{z} : z \in \{0,1\}^n\}$.
Further, we assume the matrix entries of any input local Hamiltonian are given in polar form
\begin{align*}
    H_{xy} = R_{xy}e^{i\theta(x,y)},
\end{align*}
where $R_{xy} \in \R_{\geq 0}$ and $e^{i\theta(x,y)}$ is an $m$th root of unity for some even $m \leq 2^{\poly(n)}$ specified in the input, with the exponent $k_{xy} := m\,\theta(x,y)/2\pi \in \Z_m$ computable in $\poly(n)$ time from the local description of $H$.
The product of two matrix entries $H_{x_1y_1}H_{x_2y_2}$ then has a complex phase which may be stored as an integer modulo $m$, allowing one to take the product of many entries of $H$ without increasing memory requirements.

\begin{definition}
\label{def:locality}
$H$ is a $k$\emph{-local} Hamiltonian in the computational basis if it is decomposable as a sum of terms $H = \sum_i h_i$ where $h_i$ has support on at most $k$ qubits.
If $k = O(1)$, we may simply refer to $H$ as a local Hamiltonian.
\qed
\end{definition}

\begin{definition}
    \label{def:geometric_locality}
    A Hamiltonian $H = \sum_i h_i$ is \emph{geometrically local} if its qubits can be embedded into a $p$-dimensional lattice $\Lambda$ with distance function $d_\Lambda$, where $p = O(1)$, such that there exists a constant $R = O(1)$ satisfying
    \begin{align*}
        d_\Lambda(a,b) \leq R
    \end{align*}
    for every pair of qubits $a$ and $b$ that interact through a common term $h_i$ of $H$ (i.e. $\{a,b\} \subseteq \supp(h_i)$).
    \qed
\end{definition}

\begin{definition}
\label{def:stoq}
$H$ is \emph{stoquastic} in the computational basis if
$\braket{x|H|y} \leq 0$ for all $x \neq y$.
\qed
\end{definition}

Stoquasticity is a global matrix condition that depends on the choice of basis. Some works (Ref.~\cite{Ioannou_2020}) that we reference in this paper will use the term \emph{global stoquasticity}, which is equivalent to our definition of stoquastic.

We say that a process is \emph{efficient} or can be done \emph{efficiently} if it can be performed in polynomially-many steps with respect to the size of the input.

\begin{definition}
Let $\operatorname{rank}_{\mathrm P}(A)$ denote the number of nonzero coefficients in the Pauli operator basis expansion of an operator $A$.
A family of $n$-qubit unitaries $\{U_n\}$ is (classically) \emph{Pauli-tractable} if there exists a polynomial $p$ and a classical polynomial-time algorithm that, given any Pauli string $P$, outputs the Pauli operator basis expansion of $U_n P U_n^\dagger$ with efficiently computable coefficients, and satisfies $\operatorname{rank}_{\mathrm P}(U_n P U_n^\dagger)\leq p(n)$ for every Pauli string $P$.
The family of unitaries is additionally \emph{locality-preserving} if $U_nPU_n^\dagger$ is local for every Pauli string $P$.
\qed
\end{definition}

\begin{definition}
    \label{def:eff_realizable}
    A unitary family is \emph{efficiently realizable} if it is implementable
    by a polynomial-time-uniform family of polynomial-size quantum circuits
    and is Pauli-tractable. \qed
\end{definition}

Efficiently realizable unitaries are desirable choices for finding a basis in which a given Hamiltonian is sign-problem-free.
For this reason, we will primarily consider efficiently realizable unitaries when addressing complexity-theoretic questions surrounding VGP.

\begin{definition}
    \label{def:stoq_diagonal}
    A \emph{stoquastizing  diagonal} for $H$ is a diagonal unitary $D = \sum_x e^{i\theta_x}\ket{x}\!\bra{x}$ that transforms matrix
    elements as
    \begin{align*}
        \braket{x|DHD^{\dagger}|y} = e^{i(\theta_x - \theta_y)}
        \braket{x|H|y} \leq 0, && \forall x \neq y \in \{0,1\}^n.
    \end{align*}
    In words, $D$ is a unitary transformation that renders $H$ stoquastic; in this case, we say that $D$ \emph{stoquastizes} $H$.
    \qed
\end{definition}
Diagonal elements are invariant under a diagonal unitary transformation, while off-diagonal elements $\braket{x|H|y}$ acquire a relative phase.
Stoquastizing  diagonals are precisely the unitaries that relate VGP Hamiltonians to their stoquastic counterparts.
Many of the main complexity results of this work will concern when stoquastizing diagonals are efficiently realizable or not.

\begin{definition}
    \label{def:pmr}
    The \emph{permutation matrix representation} (PMR) of a $k$-local Hamiltonian $H = \sum_{\alpha} h_{\alpha}$ is the decomposition
    \begin{align*}
        H = \sum_{i=1}^m D_iP_i,
    \end{align*}
    where $P_i$ are entry-wise disjoint permutation matrices and $D_i$ are diagonal operators. More specifically, 
    \begin{itemize}
        \item each $D_i = \sum_{j} B^{(i)}_j$ is a sum of $k$-local diagonal terms, and
        \item for each $B^{(i)}_j$ there exists $\alpha$ so that $B^{(i)}_j P_i = h_{\alpha}$.
    \end{itemize}
    We refer to each $D_iP_i$ as \emph{PMR term}. \qed
\end{definition}

Every Hamiltonian admits such a PMR decomposition~\cite{Gupta_2020}.
Additionally, if $H$ is an arbitrary spin-$s$ Hamiltonian, then the PMR terms $D_iP_i$ of $H$ may be chosen so that
\begin{align*}
    P_i = \bigotimes_{j=1}^n \Pi_{j}^{(i)}, && \left[\Pi_{j}^{(i)}\right]^d = I
\end{align*}
for $d := 2s+1$, and this choice is unique~\cite{Babakhani_2026}. We choose this as the canonical/default PMR decomposition of a Hamiltonian for the purposes of this work.
\begin{example}
    Consider $s=\frac12$.
    Then we may write $H$ as a linear dependence relation of the Pauli operator basis.
    The PMR terms of $H$ may be chosen so that $D_iP_i = \sum_{i} c_i \,Z^{\beta_i} X^{\alpha}$ for bitstrings $\alpha$ and $\beta_i$ and complex $c_i$.
    This is the canonical PMR decomposition of $H$ because every $P_i$ is its own inverse.
\qed \end{example}

\begin{definition}
    \label{def:transition_graph}
    The \emph{transition graph} $G_H$ for a Hamiltonian $H$ acting on $n$ qubits is a graph with vertex set $\{0,1\}^n$ and edge set
    \begin{align*}
        \{   (x,y)   :  x \neq y,  \braket{x|H|y} \neq 0  \}.
    \end{align*}
    We define the \emph{weight of a walk} $\ket{z_1} \to \ldots \to \ket{z_\ell}$ in $G_H$ as the product of the negation of the walk's constituent edge weights. That is, the weight satisfies
    \begin{align*}
        w\left(\,\ket{z_1} \to \ldots \to \ket{z_\ell}\,\right)    =    \prod_{i=1}^{\ell-1}- \braket{z_i|H|z_{i+1}}. && \qed
    \end{align*}
\end{definition}

The PMR framework is useful for expanding the partition function $\Tr[e^{-H}]$ into a sum over weights of closed walks in the transition graph $G_H$, because closed walks can be understood as products of PMR terms.
We formalize this notion with the following definition.

\begin{definition}
    \label{def:fundamental_cycle}
    A \emph{fundamental cycle} of a Hamiltonian $H$ is a product of PMR terms $\prod_{i=1}^L D'_iP'_i$ of $H$ such that 
    \begin{enumerate}[label=(\roman*),leftmargin=2em]
        \item $\prod_{i=1}^L P'_i = I$ and
        \item for every computational basis state $\ket{z}$ for which the constituent matrix elements are all nonzero, the closed walk in $G_H$
        \begin{equation}
            \label{eq:G_H_closed_walk}
            \ket{z} \to P'_1\ket{z} \to \ldots \to \left[\prod_{i=1}^{L-1} P'_i\right] \ket{z} \to \ket{z}
        \end{equation}
        is an \emph{induced cycle}\footnote{An \emph{induced cycle} is a closed walk with the property that any two non-consecutive vertices in the walk do not share any edges.}.\qed
    \end{enumerate}
    
\end{definition}
Every fundamental cycle encodes weights of $\leq 2^n$ closed walks. One can see that the weight of the cycle in \Cref{eq:G_H_closed_walk} is given by $\left[(-1)^L \prod_{i=1}^L D'_iP'_i\right]\ket{z}$. Each $P'_i$ generates the edge $\ket{z} \to P'_i\ket{z}$ in $G_H$ and $D'_i$ encodes the weight of that edge. The fundamental cycles provide a generating set for the space of all closed walks in $G_H$~\cite{Babakhani_2026}. The complex phases encoded in the diagonals of fundamental cycles determine the holonomies of all closed walks in $G_H$, so studying fundamental cycles is sufficient for understanding when a Hamiltonian has or does not have VGP~\cite{Hen_2021}. We will explore precisely why this is the case in \Cref{prop:vgp_equiv_def} of the next section.

\section{Vanishing Geometric Phase}
\label{sec:vgp}

We redefine VGP using the definition in \Cref{sec:intro}, introduce the stoquastic proxy map, and propose several equivalent definitions for VGP.
\subsection{Characterizations of VGP}

\begin{definition}
\label{def:vgp}
$H$ has \emph{vanishing geometric phase} (VGP) if, for every closed walk $\gamma = (z_0,z_1,\ldots,z_L = z_0)$ in $G_H$, the \emph{holonomy}
\begin{align*}
    \Phi(\gamma) := (-1)^L\prod_{t=1}^L \frac{H_{z_t,z_{t-1}}}{\left| H_{z_t,z_{t-1}}\right|}
\end{align*}
has trivial (vanishing) complex (geometric) phase. Equivalently, $\Phi(\gamma) = 1$.
A Hamiltonian that has VGP is a \emph{VGP Hamiltonian}. \qed
\end{definition}

A key observation is that every stoquastic Hamiltonian $H$ has VGP: because $H_{xy} \leq 0$, $\Phi(\gamma)$ is always 1 for any closed walk $\gamma$ in the graph of $G_H$.
Another observation is that whenever a Hamiltonian $H$ admits an acyclic transition graph $G_H$, then $H$ trivially has VGP.
Or, when $G_H$ is bipartite (it contains no odd-length cycles) and $H$ is real with $H \geq 0$ entry-wise, then $H$ has VGP.

\begin{definition}
    \label{def:stoq_proxy_map}
    Let $\mathcal{S}(\,\cdot\,)$ denote the map on matrices defined by
    \begin{align*}
        \left(\mathcal{S}(A)\right)_{ij} :=
        \begin{cases}
            A_{ii}, & i=j,\\
            -\,|A_{ij}|, & i \neq j .
        \end{cases}
    \end{align*}
    $\mathcal{S}$ is known as the \emph{stoquastic proxy map} and $\mathcal{S}(H)$ is referred to as the stoquastic proxy of $H$.
    \qed
\end{definition}

\begin{proposition}(Characterizations of VGP)
    \label{prop:vgp_equiv_def}
    The following statements are equivalent:
    \begin{enumerate}[label=(\roman*),leftmargin=2em]
        \item $H$ has VGP.
        \item For every fundamental cycle $C$ of $H$ with $L$ PMR-term factors, the matrix $(-1)^L C$ is entry-wise nonnegative in the computational basis.
        \item The weight of every induced cycle in $G_H$ has trivial complex phase.
        \item For every cycle basis\footnote{A minimal generating set of cycles for the transition graph. Cycles are generated by taking the symmetric difference of edges between cycle basis elements.} $\mathcal{C}$ of $G_H$, the weight of every $C \in \mathcal{C}$ has trivial complex phase.
        \item There exists a diagonal unitary $D$ such that $DHD^\dagger = H_{\mathrm{\stoq}}$ is stoquastic.
        \item There exists a diagonal unitary $D$ such that $\mathcal{S}(H) = DHD^\dagger$.
    \end{enumerate}
\end{proposition}

\begin{proof}
    $(ii) \iff (iii)$ is immediate by the definition of fundamental cycle (\Cref{def:fundamental_cycle}). $(i) \iff (iii) \iff (iv)$ is proven in Ref.~\cite{Hen_2021} by the fact that the complex phases of combined closed walks are additive under symmetric difference of edges. Ref.~\cite{Hen_2021} also proves $(i) \iff (v)$.

    To finish the proof, it is sufficient to show that $(v) \iff (vi)$. If $(v)$ holds, then notice $\diag(DHD^\dagger) = \diag(\mathcal{S}(H))$ since both application of $\mathcal{S}$ and conjugation by $D$ preserve the diagonal. For $x \neq y$, the unitary $D = \diag(\{e^{i \theta_z}\}_z)$ gives
    \begin{align*}
        (H_{\mathrm{\stoq}})_{xy} = e^{i(\theta_x - \theta_y)} H_{xy}.
    \end{align*}
    Since $H_{\mathrm{\stoq}}$ is stoquastic, $(H_{\mathrm{\stoq}})_{xy} \leq 0$, so $(H_{\mathrm{\stoq}})_{xy} = -|H_{xy}|$. Then
    \begin{align*}
        (\mathcal{S}(H))_{xy} = -|H_{xy}| = (H_{\mathrm{\stoq}})_{xy}.
    \end{align*}
    The off-diagonal and diagonal elements of both operators agree, so $(vi)$ follows. Assuming $(vi)$ gives $(v)$ immediately because $\mathcal{S}(H)$ is stoquastic. This completes the proof.
\end{proof}

We immediately get a useful and conceptually meaningful corollary.
\begin{corollary}\label{cor:vgp_diagonal_invariant}
    If $H$ has VGP, so does $DHD^\dagger$ for any diagonal unitary $D$.
\end{corollary}
This type of diagonal unitary invariance is the key structural property of VGP that stoquasticity lacks.
This establishes VGP as a property tied to the geometric phases of the cycles of $G_H$ and not wholly dependent on a particular basis representation.

\subsection{VGP as the Exact QMC Sign-Problem-Free Boundary}
\label{subsec:qmc_boundary}

We now state precisely the connection between VGP and
sign-problem-free quantum Monte Carlo established in Ref.~\cite{Hen_2021}, which motivates the present work.

In the PMR-QMC framework, partition function estimators take
the form of sums over closed walks $\gamma$ in $G_H$, each
weighted by a product of off-diagonal matrix elements.
When considering whether $H$ has a sign problem, the quantity of interest is the sign of the weight of each walk.
A simulation is sign-problem-free if and only if all walk
weights are nonnegative, which holds if and only if the
holonomy $\Phi(\gamma) = 1$ for every closed cycle.

\begin{theorem}[VGP is the exact QMC boundary~\cite{Hen_2021}]
\label{thm:qmc_boundary}
A Hamiltonian $H$ is sign-problem-free under PMR-QMC
\emph{if and only if} $H$ has VGP. \qed
\end{theorem}

The necessity distinguishes VGP from stoquasticity.
A stoquastic Hamiltonian trivially has VGP, so stoquasticity implies sign-problem-free simulation.
But the converse fails: there exist non-stoquastic Hamiltonians
that are VGP---they are sign-problem-free in every
diagonal-unitary-equivalent representation, even though they appear to
have positive off-diagonal elements in the computational basis.
Stoquasticity is a sufficient but not necessary condition for sign-problem-free PMR-QMC; VGP is the exact, diagonal-unitary-invariant condition.

The main complexity result of the present paper, $\VGPMA = \StoqMA$,  shows that the VGP framework is not specific to understanding the sign problem for classical simulations.
It is also a complexity boundary: diagonal unitary transformations do not enlarge $\StoqMA$ (\Cref{thm:vgp_stoqma}).

Now that we have motivated further study of the VGP framework, we present examples and properties of VGP Hamiltonians.

\subsection{Explicit VGP-but-not-stoquastic Hamiltonians}
\label{subsec:separating}

Statement $(v)$ of \Cref{prop:vgp_equiv_def} implies that there exist Hamiltonians that are VGP yet fail to be stoquastic in the computational basis.
We will provide examples of such Hamiltonians to make the separation between VGP and stoquastic Hamiltonians concrete.
These examples will also provide a minimal testbed for the stoquastic proxy construction, which will be used in the proof of one of our main results (\Cref{cor:vgpma_collapse}).

\begin{example}
    \label{ex:twoqubit_vgp}
    Consider the stoquastic Hamiltonian
    $H_0 = - X_1 - X_2$
    on two qubits, with matrix representation
    \begin{align*}
        H_0 =
        \begin{pmatrix}
            0 & -1 & -1 & 0 \\
            -1 & 0 & 0 & -1 \\
            -1 & 0 & 0 & -1 \\
            0 & -1 & -1 & 0
        \end{pmatrix}.
    \end{align*}
    The transition graph $G_{H_0}$ is the $4$-cycle
    \begin{align*}
        \ket{00} \leftrightarrow \ket{10} \leftrightarrow \ket{11}
    \leftrightarrow \ket{01} \leftrightarrow \ket{00}.
    \end{align*}
    
    Define the unitary $D = \diag(1, 1, 1, e^{-i\theta})$ for a parameter $\theta \in (0, 2\pi)$, and set $H_\theta := DH_0D^{\dagger}$.
    The matrix elements transform as $(H_\theta)_{xy} = D_{xx} \,(H_0)_{xy}\, D^\dagger_{yy}$, yielding:
    \begin{align*}
        H_\theta =
        \begin{pmatrix}
            0 & -1 & -1 & 0 \\
            -1 & 0 & 0 & -e^{i\theta} \\
            -1 & 0 & 0 & -e^{i\theta} \\
            0 & -e^{-i\theta} & -e^{-i\theta} & 0
        \end{pmatrix}.
    \end{align*}
    For $\theta \neq 0 \pmod{2\pi}$, the off-diagonal elements $(H_\theta)_{01,11} = -e^{i\theta}$ and $(H_\theta)_{10,11} = -e^{i\theta}$ are complex.
    So, $H_\theta$ is not stoquastic in the computational basis but $H_\theta$ has VGP by~\Cref{prop:vgp_equiv_def}.
    
    To see this explicitly, we compute the holonomy of the 4-cycle in $G_{H_0}$. Putting
    $\gamma = (\ket{00}, \ket{10}, \ket{11}, \ket{01}, \ket{00})$, we compute
    the holonomy 
    \begin{align}
      \Phi(\gamma)
      &= (-1)^4 \cdot \frac{(H_\theta)_{10,00}}{|(H_\theta)_{10,00}|}
      \cdot \frac{(H_\theta)_{11,10}}{|(H_\theta)_{11,10}|}\\
      &\qquad\qquad\cdot \frac{(H_\theta)_{01,11}}{|(H_\theta)_{01,11}|}
      \cdot \frac{(H_\theta)_{00,01}}{|(H_\theta)_{00,01}|} \nonumber\\
      &= (-1)^4 \cdot (-1)(-e^{-i\theta})(-e^{i\theta})(-1) = 1.
    \end{align}
    The holonomy is 1, confirming that $H_{\theta}$ indeed has VGP (\Cref{def:vgp}).
\qed \end{example}

\begin{example}[Twisted quantum doubles]
    A physically motivated source of VGP Hamiltonians comes from Ref.~\cite{Shackleton_2026}'s construction of sign-problem-free Hamiltonians for twisted quantum doubles.
    For a finite group $G$, these Hamiltonians are lattice gauge theory Hamiltonians on group-valued degrees of freedom.
    Explicitly, they take the form
    \begin{align*}
        H=-\sum_p B_p-\sum_v A_v, && A_v=\frac{1}{|G|}\sum_{g\in G} A_v^g.
    \end{align*}
    Here, $B_p$ is the diagonal projector onto configurations with no gauge flux through the plaquette $p$, while $A_v^g$ multiplies the group elements on the edges adjacent to $v$ by $g$ and attaches a phase determined by the chosen $3$-cocycle. The vertex terms are restricted to configurations with no flux on the neighboring plaquettes.
    These Hamiltonians include familiar topologically ordered models such as the double semion model, and more generally realize twisted quantum double (or Dijkgraaf--Witten) phases.
    
    We observe that these Hamiltonians have VGP.
    The diagonal terms considered in Ref. \cite{Shackleton_2026}, do not affect the transition graph (so they do not affect whether $H$ has VGP or not).
    The transition graph's off-diagonal edges come from local permutations multiplied by a phase. Let
    \begin{align*}
    \gamma=(z_0,z_1,\ldots,z_\ell=z_0)
    \end{align*}
    be any closed walk in the transition graph. Such a walk corresponds to exactly one of the nonzero operator strings appearing in the PMR closed-walk expansion---the permutation operators of the canonical PMR decomposition are matrix entry-wise disjoint, so each transition along the walk is generated by a unique PMR term---so the sign analysis of Ref.~\cite{Shackleton_2026} is precisely a computation of holonomies in the sense of \Cref{def:vgp}.
    Ref.~\cite{Shackleton_2026}'s analysis shows that every such string has vanishing geometric phase: configurations with gauge flux have no non-trivial cycles, while in the flat sector the only non-trivial cycles are global gauge transformations, whose constituent phases multiply to $1$ on a closed manifold.
    Hence
    \begin{align*}
        (-1)^\ell \prod_{j=1}^{\ell} \frac{\braket{ z_j|H|z_{j-1}}} {|\braket{ z_j|H|z_{j-1}}|} = 1.
    \end{align*}
    Thus every closed walk in the transition graph has trivial holonomy, so $H$ has VGP.
    Ref.~\cite{Shackleton_2026}'s construction therefore gives a natural, physically motivated family of local, generally non-stoquastic, VGP Hamiltonians. \qed
\end{example}

We will now present a scalable family of VGP Hamiltonians for which no efficiently realizable stoquastizing diagonal exists.

\begin{example}
\label{ex:n_qubit_family_non_local_stoquastizing _diagonal}
We provide an example of a VGP 2-local Hamiltonian
\begin{align*}
    H = Z_1 -(n-3)X_1 - X_1\sum_{i=2}^n Z_i
\end{align*}
where we take $n > 3$.
We will demonstrate that $H$ has VGP and find a stoquastizing diagonal $D$.
We will further show that the stoquastic $DHD^\dagger$ has $\operatorname{rank}_P(DHD^\dagger) = 2^{\Theta(n)}$, so $D$ is not efficiently realizable.
We will then show no other choices of $D$ can be efficiently realizable either.

Notice that $G_H$ is a partial matching of computational basis states; in particular, it contains no cycles. So $H$ trivially has VGP, but it is not stoquastic. 

The controlled-$Z$ gate $D = C^{(\otimes n)}Z_{1\ldots n}$ guarded on qubits $1, \ldots, n$ is a stoquastizing  diagonal for $H$:
\begin{align*}
    D&HD^\dagger = Z_1 -(n-3)X_1 - X_1 \sum_{i=2}^n Z_i - 4X_1 \prod_{i=2}^n \frac{I - Z_i}{2}\\
    &= Z_1 -(n-3)X_1 - X_1 \sum_{i=2}^n Z_i - 2^{3-n}X_1 \prod_{i=2}^n (I - Z_i).
\end{align*}
The Hamiltonian $DHD^\dagger$ is indeed stoquastic, but it is non-local.
In fact, the term
\begin{align*}
    -2^{3-n}X_1\prod_{i=2}^n (I - Z_i)
\end{align*}
introduces $2^{n-1}$ non-zero basis elements in the Pauli operator basis expansion of $DHD^\dagger$.
In fact, \emph{no} stoquastizing  diagonal for $H$ is efficiently realizable.
To see this, recall that conjugation by a diagonal unitary applies relative phases to existing off-diagonal matrix elements and can neither create nor destroy non-zero entries.
Every non-zero off-diagonal entry of $H$ must be mapped to the negation of its absolute value, so \emph{every} stoquastizing  diagonal $D'$ satisfies $D'H(D')^\dagger = \mathcal{S}(H) = H - 4X_1 \otimes \ket{1^{n-1}}\bra{1^{n-1}}$ exactly.
Since the Pauli-basis expansion of $\mathcal{S}(H)$ is unique and contains $2^{n-1}$ terms, no stoquastizing  diagonal maps the local description of $H$ to a polynomial-size local description; hence none is efficiently realizable.
We note that $C^{(\otimes n)}Z_{1\ldots n}$ itself has a polynomial-size quantum circuit and a very simple phase function ($\theta_z = \pi$ exactly on the all-ones string); what fails is the possibility of having a compact representation of $\mathcal{S}(H)$.
\qed \end{example}

\subsection{VGP for Weighted Hypercube with Disjoint PMR Supports}

\begin{proposition}
    \label{prop:disjoint_eff_vgp_recognition}
    Suppose $H = H_0 + \sum_i D_iX_i$ is a local real Hamiltonian decomposed into PMR form, with $H_0$ diagonal.
    If
    \begin{align}
        \label{eq:disjoint_pmr}
        \left[\bigcup_i \supp(D_i)\right] \, \cap \, \left[\bigcup_i \supp(X_i)\right] = \emptyset,
    \end{align}
    then $H$ has VGP.
\end{proposition}

\begin{remark}
    One can efficiently recognize when a local Hamiltonian satisfies the assumptions of~\Cref{prop:disjoint_eff_vgp_recognition} and certify that it has VGP. \qed
\end{remark}
\begin{example}
    A Hamiltonian of this form need not be stoquastic.
    On $n+m$ qubits, write
    \begin{align*}
        H = \sum_{i=1}^n h_{i} X_i, && \supp(h_i) \cap [n] = \emptyset, && h_i \geq 0,
    \end{align*}
    where each $h_i$ is a diagonal operator acting on the last $m$ (spectator) qubits, entry-wise nonnegative, and with at least one strictly positive entry.
    The strictly positive entries of $h_i$ make $h_i X_i$ non-stoquastic, but $H$ meets the assumptions of~\Cref{prop:disjoint_eff_vgp_recognition}. \qed
\end{example}
\begin{proof}[Proof (of~\Cref{prop:disjoint_eff_vgp_recognition})]
    Two key facts following from the assumptions:
\begin{itemize}
    \item[$(i)$] The transition graph only consists of single-qubit $X$ (single bit flip) transitions, hence to form a fundamental cycle $F$, any PMR term of $F$ must appear as a factor of $F$ an even number of times.
    \item[$(ii)$] The weight of each $D_pX_p$ transition is completely independent of all qubits $q$ such that there is a PMR term $D_qX_q$ of $H$.
    Consequently, in a connected component $C \subseteq G_H$, the weight of the edge induced by the transition $D_pX_p\ket{z}$ is independent of choice of $\ket{z} \in C$. 
\end{itemize}
By the first fact, if a transition $D_pX_p$ is included in a fundamental cycle $F$, it must repeat an even number of times in $F$.
By the second fact, this means geometric phase contributed by the repeated factors of $D_pX_p$ in $F$ vanishes (i.e. $\equiv 0$), since $H$ is real.
Hence $H$ has VGP by the fundamental cycle characterization of VGP from \Cref{prop:vgp_equiv_def}.
\end{proof}

This provides an efficient way to screen for VGP in local Hamiltonians (with one-sided error).
Remarkably, as soon as one $D_i$ is allowed to be an arbitrary Ising Hamiltonian, recognizing stoquasticity in $H$ becomes $\coNP$-hard.
To see this, first recall that $D_iX_i$ is matrix entry-wise disjoint from the other PMR terms. 
Secondly, recall that deciding if an Ising Hamiltonian has a diagonal matrix entry $> 0$ is $\NP$-hard~\cite{Lucas_2014}.
The converse problem is $\coNP$-hard, so deciding if all of the matrix entries contributed by $D_iP_i$ are non-positive is $\coNP$-hard.
Thus deciding if $H$ is stoquastic is $\coNP$-hard in this case.

Considering the preceding discussion, a noteworthy consequence of~\Cref{prop:disjoint_eff_vgp_recognition} can be summarized in the following corollary:
\begin{corollary}
    \label{cor:recognition_vgp_easy_stoq_hard}
    There exist local Hamiltonians where recognition of the VGP property is efficient, but recognition of stoquasticity is $\coNP$-hard. \qed
\end{corollary}
As a quick remark, off-diagonal Hamiltonians $H$ meeting the assumptions in \Cref{prop:disjoint_eff_vgp_recognition} are easy to diagonalize by applying a Hadamard to every qubit $i$ such that $D_iX_i$ is a PMR term of $H$.
However, as soon as $H_0$ contains diagonal terms such as $D_iZ_i$ then this is no longer necessarily the case.

\subsection{VGP Hamiltonians that are Hard to Stoquastize}

We aim to further solidify the separation between VGP and stoquastic Hamiltonians through a complexity-theoretic proof.
\begin{theorem}
    \label{thm:vgp_hard_to_stoquastize}
    There are VGP 3-local Hamiltonians $H$ that simultaneously satisfy:
    \begin{enumerate}
        \item[(i)] No efficiently realizable stoquastizing  diagonal for $H$ exists.
        \item[(ii)] Deciding if a layer of Hadamard gates can stoquastize $H$ is $\NP$-hard.
        \item[(iii)] $H$ belongs to a family of Hamiltonians that admit efficient recognition of the VGP property.
    \end{enumerate}
\end{theorem}

Ref.~\cite{Marvian_2019} establishes $(ii)$ as a standard for demonstrating the hardness of stoquastizing  a family of local Hamiltonians.
Thus showing $(i)$ and $(ii)$ demonstrates that there are VGP local Hamiltonians that are provably hard to stoquastize, and showing $(iii)$ establishes that the preceding properties do not contradict efficient recognition of VGP.

Before proceeding with the construction, we first review a related result from Ref.~\cite{Marvian_2019}.
Ref.~\cite{Marvian_2019} proves that, given a 3-local Hamiltonian $H$, deciding if a layer of Hadamard gates stoquastizes $H$ is an $\NP$-hard task.
They accomplish this by constructing a bijection between the clauses $C$ of a Boolean formula $\psi$ in conjunctive normal form---with at most 3 literals per clause---and a 3-qubit Hamiltonian gadget $H_C$.
Each $H_C$ is stoquastized after application of a Hadamard to at least one of the qubits of $\supp(H_C)$, and $H = \sum_C H_C$ is constructed in such a way that each $H_C$ can be simultaneously stoquastized by a choice of Hadamards if and only if $\psi$ is satisfiable.
The satisfiability of $\psi$ is $\NP$-hard to decide (this is known as the 3-SAT problem)~\cite{Garey_1979}.
However, the Hamiltonian that Ref.~\cite{Marvian_2019} constructs does not have VGP in general.
We aim to extend their result by constructing a VGP Hamiltonian with the same property, plus the additional properties $(i)$ and $(iii)$.

\begin{proof}[Proof (of~\Cref{thm:vgp_hard_to_stoquastize})]
    Fix a 3-SAT instance with Boolean formula $\psi$ in conjunctive normal form on variables $X = \{x_1, \ldots, x_n\}$ and $n \geq 2$.
    Let $\mathcal{C}$ be the set of clauses of $\psi$.
    For every variable $x_i$ appearing in $\psi$, introduce two ``Boolean assignment'' qubits $p_i$ and $n_i$.
    Let $q(x_i) = p_i$ and $q(\neg x_i) = n_i$.
    For every clause $C \in \mathcal{C}$ with literals $(\ell_i \vee \ell_j \vee \ell_k)$, introduce an ancilla qubit $a_C$.
    Construct the gadget Hamiltonian
    \begin{align*}
        H_C := Z_{a_C} - 2X_{a_C}  - X_{a_C} (Z_{q(\ell_i)} + Z_{q(\ell_j)} + Z_{q(\ell_k)})
    \end{align*}
    so that if at least one qubit $q(\ell_{(\,\cdot\,)})$ has a Hadamard applied to it by some layer of Hadamards $U$, then $U(H_C)U^\dagger$ is stoquastic.
    Notice that if $U$ applies a Hadamard gate to $a_C$, then $+Z_{a_C} \mapsto +X_{a_C}$ introduces a non-stoquastic term under $U$.
    
    Now we introduce a separate gadget $H_i^{(1)}$ that forces exactly one Hadamard to be applied on either $p_i$ or $n_i$ for every $H_i^{(1)}$ to be simultaneously stoquasticizable by Hadamards.
    For every variable $x_i$, introduce two new qubits $u_i$ and $v_i$ and add the following term
    \begin{align*}
        H_i^{(1)} &:= H_i^{(\geq1)} + H_i^{(\leq1)},\\
        H_i^{(\geq1)}&=[Z_{u_i} - X_{u_i} - X_{u_i}(Z_{p_i} + Z_{n_i})] \\
        H_i^{(\leq1)}&= [Z_{v_i} + X_{v_i}(-3 -Z_{p_i} - Z_{n_i} + Z_{p_i}Z_{n_i})].
    \end{align*}
    The term $H_i^{(\geq1)}$ requires that at least one of $p_i$ and $n_i$ have a Hadamard applied to them for $H_{i}^{(1)}$ to be stoquastized, while the term $H_i^{(\leq1)}$ enforces that at most one Hadamard may be applied to $p_i$ and $n_i$ for $H_{\text{i}}^{(1)}$ to be stoquastized.
    Additionally, applying a Hadamard to qubits $u_i$ or $v_i$ forces $H_i^{(1)}$ to be mapped to a non-stoquastic Hamiltonian, since the diagonal terms map as $+Z_{u_i} \mapsto +X_{u_i}$ and $+Z_{v_i} \mapsto +X_{v_i}$.
    So applying a Hadamard to these qubits is not helpful for the decision problem.

    Next, introduce a qubit $s$ and put
    \begin{align*}
        H_\omega &:=  - (2n-2)X_s -  X_s\sum_{x_i \in X} (Z_{p_i} + Z_{n_i})
    \end{align*}
    The Hamiltonian $H_\omega$ is stoquastized by a layer of Hadamards $U$ whenever $U$ applies Hadamards to at least two of the $2n$ qubits from $\{p_i, \, n_i\}_{x_i \in X}$ and to no other qubit in $\supp(H_\omega)$: with Hadamards on $r$ of the Boolean assignment qubits, the worst-case coefficient of a single flip of $s$ is $-(2n-2) + (2n - r) = 2 - r \leq 0$, while each Hadamarded control contributes only a $-X_sX_{q}$ term with non-positive entries.

    Finally, to complete the reduction, construct
    \begin{align*}
        H := H_\omega + \sum_{C \in \mathcal{C}} H_C + \sum_{x_i \in X} H_i^{(1)}.
    \end{align*}
    To recapitulate, $H_C$ requires that a Hadamard is applied to at least one qubit $q(\ell_j)$ where the clause is $C = (\ell_i \vee \ell_j \vee \ell_k)$.
    $H_i^{(1)}$ requires that exactly one of $p_i$ or $n_i$ has a Hadamard applied to it (analogous to the Law of Excluded Middle).
    $H_\omega$ is stoquastized whenever at least two qubits from $\{p_i, \, n_i\}_{x_i \in X}$ have Hadamards applied to them (and none of its other qubits do), which is always the case in a satisfying assignment.
    Thus $\psi$ is satisfiable if and only if $H$ can be simultaneously stoquastized by a layer of Hadamards.
    Construction of $H$ is efficient, so property $(ii)$ holds.
    
    $H_\omega$ seems redundant, but it is what encodes the impossibility of an efficiently realizable stoquastizing diagonal for $H$.
    The argument of \Cref{ex:n_qubit_family_non_local_stoquastizing _diagonal} applies verbatim to $H_\omega$ (now with $2n$ control qubits): $G_{H_\omega}$ is a partial matching, the single positive entry sits at the all-ones assignment of the controls, and every stoquastizing  diagonal is forced to produce $\mathcal{S}(H_\omega) = H_\omega - 4X_s \otimes \ket{1^{2n}}\!\bra{1^{2n}}$, whose unique Pauli expansion contains $2^{2n}$ terms; hence no efficiently realizable stoquastizing  diagonal exists for $H_\omega$.
    Recall that conjugation by a diagonal unitary only applies relative phases to individual off-diagonal matrix elements; it does not create or destroy non-zero matrix entries.
    Further, see that the matrix entries of $H_\omega$ do not overlap with those of any $H_C$ or $H_i^{(1)}$.
    So, a stoquastizing diagonal for $H$ would also be a stoquastizing diagonal for $H_\omega$.
    Hence no efficiently realizable stoquastizing  diagonal can exist for $H$ either.

    We remark that applying a Hadamard to $s$ leaves properties $(i)$ and $(ii)$ intact.
    Such a Hadamard makes $H_\omega$ diagonal (hence trivially stoquastic), but $s$ occurs in no other term, so the gadgets $H_C$ and $H^{(1)}_i$ are unchanged and $H$ is not stoquastized.
    Property $(i)$ concerns diagonal unitaries, and a Hadamard is not one; and in property~$(ii)$, any satisfying assignment already applies Hadamards to at least two of the $\{p_i, n_i\}$---stoquastizing $H_\omega$ as intended---while a Hadamard on $s$ never helps satisfy the clause or variable gadgets.
    Thus a Hadamard on $s$ is an always-available but inconsequential stoquastization of $H_\omega$ alone.

    Lastly, $H$ satisfies the assumptions of~\Cref{prop:disjoint_eff_vgp_recognition}, so it admits efficient recognition of the VGP property.
    This completes the proof.
\end{proof}
This result serves to demonstrate that the perspective of stoquasticity is too narrow.
Indeed, it is easy to verify that $H$ has VGP, but it is hard to stoquastize it---either by searching for an efficiently realizable diagonal unitary or an appropriate layer of Hadamards.
But the VGP promise---as we will find in~\Cref{sec:vgpma}---is all you need for ground-state energy estimation by a $\StoqMA$ verifier.
Therefore, it may be superfluous to attempt to stoquastize a non-stoquastic local Hamiltonian that obviously has VGP. 
This holds true whether your goal is a sign-problem-free classical simulation via PMR-QMC (\Cref{thm:qmc_boundary}) or an analysis of the ground-state on a quantum computer.

However, it can still be useful to consider whether a VGP Hamiltonian $H$ can be stoquastized by some convenient choice of unitary.
For example, if one can prove the existence of a stoquastizing  diagonal for $H$, then they have proven that $H$ has VGP (by \Cref{prop:vgp_equiv_def}).
We will explore the tractability of this task later in~\Cref{sec:diagonal_recovery}.

\subsection{Perron--Frobenius Structure for VGP Hamiltonians}

One of the fundamental structural consequences of stoquasticity is the Perron--Frobenius property: on each connected component of the transition graph, the ground-state may be chosen entry-wise nonnegative in the computational basis.
This subsection establishes that the Perron--Frobenius property extends exactly to VGP Hamiltonians up to diagonal unitary transformation, providing the invariant analog of the Perron--Frobenius theorem and clarifying the structure of ground-states in the VGP class.

\begin{theorem}[VGP Perron--Frobenius Structure]
\label{thm:vgp_pf}
Let $H$ be a VGP Hamiltonian with stoquastizing  diagonal
$D = \sum_z e^{i\theta_z}\ket{z}\!\bra{z}$.
On each connected component $C$ of $G_H$, there exists a ground-state $\ket{\psi_0}$ of the restriction $H|_C$ (a ground-state of $H$ itself whenever $C$ attains $\lambda_{\min}(H)$) whose amplitudes take the form
\begin{align*}
  \braket{z|\psi_0} = e^{-i\theta_z}\, a_z, && a_z \ge 0, \quad z \in C.
\end{align*}
Equivalently, $D\ket{\psi_0}$ can be chosen entry-wise nonnegative in the computational basis.
Moreover the ground energy of $H\vert_C$ is nondegenerate, $a_z > 0$ for all $z \in C$, and the phase pattern $\{e^{-i\theta_z}\}$ is unique up to a global phase; degeneracy of $\lambda_{\min}(H)$ can therefore arise only from distinct connected components attaining the same energy.
\end{theorem}

\begin{proof}
Since $H_{\stoq} := DHD^{\dagger}$ is stoquastic, the matrix $cI - H_{\stoq}$ is entry-wise nonnegative on each connected component for $c \geq \max_z \braket{z|H|z}$.
By the Perron--Frobenius theorem for nonnegative matrices~\cite{Horn_Johnson_2013}, its spectral radius is an eigenvalue with a nonnegative eigenvector---a lowest-energy state of $H_{\stoq}$ restricted to the connected component; the eigenvector of any connected component is entry-wise positive and unique up to scaling.
Translating back: on each connected component $C$, there exists a ground-state $\ket{\phi_0}$ of $H_{\stoq}\vert_C$ with $\braket{z|\phi_0} \ge 0$ for all $z \in C$.

Since $H$ and $H_{\stoq} = DHD^{\dagger}$ are unitarily equivalent, defining $\ket{\psi_0} := D^{\dagger}\ket{\phi_0}$ yields a ground-state of $H\vert_C$ with amplitudes
\begin{align*}
  \braket{z|\psi_0} = e^{-i\theta_z}\braket{z|\phi_0} = e^{-i\theta_z} a_z, && a_z \ge 0.
\end{align*}
The nondegeneracy, positivity, and uniqueness claims follow from the simplicity of the Perron root on each connected component, applied to $H_{\stoq}$.
\end{proof}

Because diagonal conjugation by $D$ preserves eigenvalues and maps eigenvectors by a fixed phase pattern, any spectral result for stoquastic Hamiltonians that relies solely on Perron--Frobenius structure extends in principle to VGP Hamiltonians after conjugation by $D$.

However, as we demonstrated in \Cref{ex:n_qubit_family_non_local_stoquastizing _diagonal}, $D$ may not be efficiently realizable, so conjugation by $D$ may not be feasible in general.

\section{PMR Locality \& Efficient stoquastization}
\label{sec:pmr_locality}
\label{sec:efficient_stoquastization}
In this section, we introduce PMR locality, which is a notion of locality rooted in the PMR framework (\Cref{def:pmr}).
We then demonstrate that PMR locality is a weaker restriction than geometric locality but stronger than $k$-locality.
We provide examples of Hamiltonians that are of interest in the literature to make the separation between these notions of locality concrete (\Cref{ex:central_spin_model}, \Cref{ex:pmr-local_not_k-local_yes}).
We then show that PMR-local Hamiltonians that are promised to have VGP can be efficiently transformed into a stoquastic local Hamiltonian (\Cref{cor:efficient_stoquastization_pmr_local}).
We will use this important feature of PMR local Hamiltonians several times throughout this work; particularly when discussing the complexity of decision problems related to VGP and its relationship to stoquasticity.

\subsection{PMR Locality}

\begin{definition}
    \label{def:pmr_locality}
    A Hamiltonian $H = H_0 + \sum_i D_iP_i$ cast in PMR form is \emph{$k$-PMR-local} if $|\supp(D_iP_i)| \leq k$ for every $D_iP_i$.
    If $k = O(1)$, we may say that $H$ is \emph{PMR-local}. \qed
\end{definition}
\noindent We emphasize that the PMR locality condition holds \emph{only} for the purely off-diagonal PMR terms of $H$. In general, $\supp (H_0)$ may scale linearly with $n$ but $H$ can still be PMR-local.
\begin{remark}
    \label{rem:pmr_local_poly_many_entries}
    Every PMR-local $H$ has $\poly(n)$ distinct off-diagonal matrix-element entries.
    So, not only do PMR-local Hamiltonians offer us efficient computability of the matrix entries $H_{xy}$ (as do any local Hamiltonian), but we are guaranteed that there are only polynomially many distinct off-diagonal matrix-element values.
    \qed
\end{remark}

PMR locality is the premise of several of this work's complexity results regarding VGP and its relationship to stoquasticity.
The above definitions will enable us to prove that
\begin{enumerate}
    \item PMR-local Hamiltonians that are promised to have VGP can be efficiently transformed into a cospectral stoquastic Hamiltonian (\Cref{cor:efficient_stoquastization_pmr_local}).
    \item One can efficiently decide if a local block-diagonal stoquastizing  diagonal exists for PMR-local Hamiltonians (certifying VGP in the YES case) (\Cref{thm:dstoq_efficient}).
    \item The frustration-free local Hamiltonian problem is in $\MA$ for PMR-local Hamiltonians that are promised to have VGP (\Cref{thm:ff-pmr-local-lh-ma}).
\end{enumerate}

Before we proceed with proving these results, we would like to provide further motivation for the use of PMR locality.
We will prove that PMR locality sits neatly between two well-established notions of Hamiltonian locality.
\begin{proposition}
    \label{prop:geo_loc_implies_pmr_loc}
    If $H = \sum_{\alpha \in A} h_\alpha$ is geometrically local, then $H$ is PMR-local.
\end{proposition}
\begin{proof}
    Each PMR term $D_iP_i$ of $H$ may be decomposed as
    \begin{align*}
        D_iP_i = \sum_{\beta \in B_i} h_\beta
    \end{align*}
    for some $B_i \subset A$. 
    By definition of geometric locality (\Cref{def:geometric_locality}), there exists an $O(1)$-dimensional lattice $\Lambda$ over $\supp(H)$ so that for every $\alpha$,
    \begin{align*}
        \{i,j\} \subseteq \supp(h_\alpha) && \implies &&d_\Lambda(i,\, j) = O(1).
    \end{align*}
    Then
    \begin{align*}
        \left|\supp(D_iP_i)\right| = \left|\bigcup_{\beta \in B_i} \supp(h_\beta)\right| = O(1).
    \end{align*}
    So, $H$ is PMR-local.
\end{proof}
\noindent The converse of \Cref{prop:geo_loc_implies_pmr_loc} is not true. To make this separation explicit, we provide an example of a well-studied model in quantum information theory that fails to be geometrically local, but is indeed PMR-local.
\begin{example}
    \label{ex:central_spin_model}
    The Central Spin Model is fundamental to understanding decoherence in a large class of potential realizations of qubits~\cite{Uhrig_2014}. Its Hamiltonian 
    \begin{align*}
        H = \sum_{k=2}^n J_k \left( X_1 X_k + Y_1Y_k + Z_1Z_k \right)
    \end{align*}
    is not geometrically local. However, if we express $H$ in PMR form:
    \begin{align*}
        H = \sum_{k=2}^n J_k\,Z_1Z_k + \sum_{k=2}^n \underbrace{J_k(I - Z_1Z_k)}_{=:\,D_k} \underbrace{X_1X_k}_{=: \,P_k},
    \end{align*}
    then $|\supp(D_kP_k)| = 2$. So, $H$ is PMR-local.
\qed \end{example}

\noindent All $k$-PMR-local Hamiltonians are trivially $k$-local, but, as the next example demonstrates, not all $k$-local Hamiltonians are PMR-local.
\begin{example}
    \label{ex:pmr-local_not_k-local_yes}
    The Hamiltonian
    \begin{align*}
        H = \sum_{i=2}^n Z_iX_1 
    \end{align*}
    is 2-local, but it has a single PMR term with support on $n$ qubits.
    So, $H$ is not PMR-local.
\qed \end{example}

\noindent Thus the hierarchy between these notions of locality is
\begin{center}
    geometric locality $\subset$ PMR locality $\subset$ $k$-locality.
\end{center}
\begin{remark}
    All results that we prove for PMR-local Hamiltonians extend to geometrically local Hamiltonians.
    \qed
\end{remark}

One useful property of PMR locality is that it admits efficient recognition of stoquasticity, which will be relevant to our discussion of the complexity of recognizing VGP in PMR-local Hamiltonians.
\begin{proposition}
    \label{prop:efficient_stoq_check_pmr_local}
    If $H = \sum_{i=1}^m D_iP_i$ is $k$-PMR-local, then one can decide if $H$ is stoquastic in time $O(2^k m)$.
\end{proposition}
\begin{proof}
    By $k$-PMR locality, each $D_iP_i$ takes on at most $2^k$ non-zero matrix entries. Additionally, for every $i \neq j$, $D_iP_i$ and $D_jP_j$ are matrix entry-wise disjoint.
    Therefore, to check for stoquasticity of $H$, one can check, for every off-diagonal PMR term $D_iP_i$ and every assignment $z \in \{0,1\}^k$ of the qubits in $\supp(D_iP_i)$, whether $\bra{P_i\, z}D_iP_i\ket{z} \leq 0$ (diagonal entries are unconstrained by stoquasticity).
    This gives a total run time of $O(2^k m)$.
\end{proof}

\subsection{Efficient Stoquastization of PMR-local Hamiltonians Under a VGP Promise}

We introduce the following procedure for efficiently constructing the stoquastic proxy $\mathcal{S}(\, \cdot\,)$ (\Cref{def:stoq_proxy_map}) on PMR-local inputs.
\begin{proposition}
    \label{prop:efficient_stoquastization_pmr_local}
    Assume $H = H_0 + \sum_{i=1}^m D_iP_i$ is $k$-PMR-local. Then one can obtain a $k$-PMR-local Pauli-basis expansion of $\mathcal{S}(H)$ in time $O(2^k m)$.
\end{proposition}
\begin{proof}
    Since the PMR terms of $H$ are entry-wise disjoint, we have
    \begin{align*}
        \mathcal{S}(H) = H_0 + \sum_{i=1}^m \mathcal{S}(D_iP_i).
    \end{align*}
    It is sufficient to show that the Pauli basis expansion of $\mathcal{S}(D_iP_i)$ is computable in time $O(2^k)$.
    
    Put $S_i := \supp(D_iP_i)$. Define the new operator $(D_iP_i)^{(\stoq)}$ on inputs $(x,y) \in \left(\{0,1\}^{|S_i|}\right)^2$ as
    \begin{align*}
        (D_iP_i)^{(\stoq)}_{xy} :=
            -|(D_iP_i)_{xy}|.
    \end{align*}
    Since $|S_i| = k$ and $D_iP_i$ has at most $2^k$ non-zero matrix entries, $(D_iP_i)^{(\stoq)}$ is constructible in $2^k$ steps. Similarly, its Pauli-basis expansion is computable in $O(2^k)$ time. 
    
    Notice that $(D_iP_i)^{(\stoq)} = \mathcal{S}(D_iP_i)$.
    So, a Pauli-basis expansion for $\mathcal{S}(H)$ is computable in time 
    \begin{align*}
        \sum_{i=1}^m 2^{|S_i|} \leq 2^km = O(2^k n^k)
    \end{align*}
    since locality gives $m = O(n^k)$. By the above construction, the support of $D_iP_i$ does not grow under the map $\mathcal{S}$. Hence $\supp (\mathcal{S}(D_iP_i)) \subseteq \supp (D_iP_i)$, so $\mathcal{S}(H)$ is $k$-PMR-local.
\end{proof}
\begin{corollary}
    \label{cor:efficient_stoquastization_pmr_local}
    If a PMR-local Hamiltonian $H$ is promised to have VGP, then one can efficiently construct a stoquastic PMR-local Hamiltonian $\mathcal{S}(H)$ that is cospectral to $H$.
\end{corollary}
\begin{proof}
    Efficient construction of the PMR-local Pauli-basis expansion of $\mathcal{S}(H)$ is~\Cref{prop:efficient_stoquastization_pmr_local}. By~\Cref{prop:vgp_equiv_def}, since $H$ has VGP there is a diagonal unitary $D$ with $\mathcal{S}(H) = DHD^\dagger$; hence $\mathcal{S}(H)$ is cospectral to $H$.
\end{proof}
In other words, the above corollary states that a PMR-local $H$ may be efficiently stoquastized so long as it is promised to have VGP.
An immediate consequence of this is that the local Hamiltonian ground-state energy problem for VGP-promised PMR-local Hamiltonians is in $\StoqMA$~\cite{Bravyi_2006}. 
We will touch on this point again in the next section, \Cref{sec:vgpma}, when we prove that the same holds for \emph{all} local Hamiltonians promised to have VGP (\Cref{thm:vgp_stoqma}).

We emphasize that constructing a stoquastic proxy for a PMR-local $H$ is efficient when $H$ is \emph{promised to have} VGP. 
A natural follow-up question is: what complexity classification does VGP recognition fall under? 
In~\Cref{sec:vgp_complexity}, we will address this question for the PMR-local and $k$-local cases.

We now move on to present the main result of the paper.

\section{The Complexity Class $\VGPMA$}
\label{sec:vgpma}

We introduce some preliminary definitions and existing results that are necessary for our discussion of Hamiltonian complexity.
For the precise definition of $\StoqMA$, see Appendix~\ref{app:complexity_class_definitions}.

\begin{definition}
\label{def:local_hamiltonian_problem}
The local Hamiltonian problem is as follows:
    \begin{itemize}[leftmargin=5.5em]
        \item[\textbf{Input}:] $H$ is a local Hamiltonian, thresholds $a < b$ with $b - a \geq 1/\poly(n)$.
        \item[\textbf{Promise}:] Either $\lambda_{\min}(H) \leq a$ or $\lambda_{\min}(H) \geq b$.
        \item[\textbf{Problem}:] Decide whether $\lambda_{\min}(H) \leq a$ or $\lambda_{\min}(H) \geq b$. \qed
    \end{itemize}
\end{definition}
Denote by $\VGPLH$ and $\STOQLH$ the local Hamiltonian problems under the promises that the input Hamiltonians are VGP and stoquastic, respectively.
We will now define the class $\VGPMA$.
\begin{definition} \label{def:vgpma}
Define $\VGPMA$ as the class of decision problems that are polynomial-time many-one reducible to $\VGPLH$. \end{definition}

We immediately get the following containments:

\begin{proposition}
\label{prop:containments}
$\StoqMA \subseteq \VGPMA \subseteq \QMA$. \qed
\end{proposition}

The goal of this section is to prove that $\StoqMA = \VGPMA$.
In particular, we will prove that $\VGPLH \in \StoqMA$, which immediately yields the desired collapse.

\subsection{A proof that $\VGPLH \in \StoqMA$}
\label{sec:collapse}

The main goal of this section is to prove 
\begin{theorem}
    \label{thm:vgp_stoqma}
    $\VGPLH \in \StoqMA$. \qed
\end{theorem}
This will immediately yield
\begin{corollary}
\label{cor:stoqma_complete}
    $\VGPLH$ is $\StoqMA$-complete.
\end{corollary}
\begin{proof}
    The stoquastic local Hamiltonian problem is $\StoqMA$-hard and stoquastic $\implies$ VGP, so $\VGPLH$ is $\StoqMA$-hard.
    Thus showing $\VGPLH \in \StoqMA$ yields the corollary.
\end{proof}

To prove Theorem~\ref{thm:vgp_stoqma}, we model our approach on the proof of Theorem~1 from Ref.~\cite{Ioannou_2020} (Ioannou et al.), which establishes that the local Hamiltonian problem for stoquastic local Hamiltonians is contained in $\StoqMA$\footnote{Ref.~\cite{Bravyi_2006} originally proved this for \emph{termwise stoquastic} Hamiltonians (see \Cref{def:termwise_stoquastic}), while Ref.~\cite{Ioannou_2020} was the first to prove it for \emph{globally stoquastic} Hamiltonians. Our definition of stoquastic is equivalent to being globally stoquastic.}.

\subsubsection{The Stoquastic Local Hamiltonian Problem is $\StoqMA$-complete}
\label{subsec:glob_stoq}

We review known results for $\STOQLH$.
\begin{theorem}[Theorem 1 of~\cite{Ioannou_2020}]
    \label{thm:glob_stoq_stoqma}
    $\STOQLH \in \StoqMA$. \qed
\end{theorem}

\begin{corollary}[Corollary 1.1 of~\cite{Ioannou_2020}]
    \label{cor:glob_stoq_stoqma}
    $\STOQLH$ is $\StoqMA$-complete. \qed
\end{corollary}

The proof of Theorem~1 in Ref.~\cite{Ioannou_2020} uses a lemma which decomposes a stoquastic $k$-local Hamiltonian $H$ into a form that lends itself to ground-state energy estimation by a $\StoqMA$ verifier. 

\begin{lemma}[Lemma 2 of~\cite{Ioannou_2020}]
\label{lem:glob_stoq_decomp}
Let 
$
H = \sum_{i=1}^m h_i
$
be a stoquastic $k$-local Hamiltonian. Then there exists $\beta < 0$ such that
\begin{equation}
\label{eq:stoq_decomposition}
H + \beta I
=
- H_0
+
\sum_{j=1}^{m'}
U_j (-X \otimes H_j) U_j^\dagger
\end{equation}
where $U_j$ is a quantum circuit of $X$ and CNOT gates and $H_j$ is a classical (diagonal) Hamiltonian with $H_j \ge 0$ for $j \ge 0$. The number of terms in the sum $m' \le m 2^{2k}.$ One can (classically) efficiently find this decomposition, i.e.\ determine $\beta$, $U_j$ and the description of $H_j$ for $j \ge 0$.
\end{lemma}

Once this decomposition is acquired, the property of stoquasticity is no longer used in the proof of Theorem 1 in Ref.~\cite{Ioannou_2020}. 
Their proof only requires that $H$ is efficiently decomposable into the form in~\Cref{eq:stoq_decomposition} and that each diagonal $H_j$ has efficiently computable diagonal entries, not necessarily that $H$ is stoquastic. 
Therefore, to prove Theorem~\ref{thm:vgp_stoqma}, it is sufficient to show that for any VGP local Hamiltonian $H$, we can efficiently perform the same decomposition for some $H'$ that is spectrally equivalent to $H$.

\subsubsection{Decomposing the Stoquastic Proxy}

Recall the stoquastic proxy map $\mathcal{S}( \, \cdot \,)$ (\Cref{def:stoq_proxy_map}), which provides a unitarily-similar, stoquastic Hamiltonian on a VGP Hamiltonian input (\Cref{prop:vgp_equiv_def}).
Hence to prove Theorem~\ref{thm:vgp_stoqma}, it is sufficient to show that we can efficiently perform the same decomposition as in Lemma~\ref{lem:glob_stoq_decomp} for $\mathcal{S}(H)$ given a local Hamiltonian $H$.

\begin{lemma}[Extension of \Cref{lem:glob_stoq_decomp}]
\label{lem:stoq_map_decomp}
Let 
$
H = \sum_{i=1}^m h_i
$
be a $k$-local Hamiltonian. Then there exists $\beta < 0$ such that
\begin{equation}
\label{eq:S_decomposition}
\mathcal{S}(H) + \beta I
=
- H_0
+
\sum_{j=1}^{m'}
U_j (-X \otimes H_j) U_j^\dagger,
\end{equation}
where $U_j$, $H_j$ and $m'$ are as defined in Lemma~\ref{lem:glob_stoq_decomp}, and one can efficiently find this decomposition classically.
\end{lemma}

\begin{proof}
The decomposition of Lemma~\ref{lem:glob_stoq_decomp} (Lemma~2
of~\cite{Ioannou_2020}) uses the Hamiltonian only through the family
of diagonal operators $H_{S,x}$, indexed by qubit subsets
$S \subseteq [n]$ with $|S| \le k$ and bitstrings
$x \in \{0,1\}^{|S|}$.
Ref.~\cite{Ioannou_2020} defines
\begin{align}
\label{eq:HSx_original}
   H_{S,x} = -\sum_{y \in \{0,1\}^{n - |S|}}  \bra{xy}\, H\, \ket{\bar{x}y}\,\ket{y}\bra{y}.
\end{align}
We define the corresponding family for $\mathcal{S}(H)$:
\begin{align}
\label{eq:HSx_modified}
   \widetilde{H}_{S,x} &= -\sum_{y \in \{0,1\}^{n - |S|}}  \bra{xy}\, \mathcal{S}(H)\, \ket{\bar{x}y}\,\ket{y}\bra{y} \nonumber\\
   &= \sum_{y \in \{0,1\}^{n - |S|}} \bigl|\bra{xy}\, H\, \ket{\bar{x}y}\bigr| \,\ket{y}\bra{y},
\end{align}
where the step \eqref{eq:HSx_modified} uses
$\bra{xy}\,\mathcal{S}(H)\,\ket{\bar{x}y} = -|\bra{xy}\,H\,\ket{\bar{x}y}|$
for $x \ne \bar{x}$.

We verify the three properties required by the proof
in Ref.~\cite{Ioannou_2020}:
\begin{enumerate}[label=(\alph*),leftmargin=2em]
\item \emph{Diagonal and nonnegative.}
Each $\widetilde{H}_{S,x}$ is diagonal in the
computational basis and has nonnegative coefficients, since
each summand in~\eqref{eq:HSx_modified} is an absolute value.

\item \emph{Efficient computability.}
For each $(S, x, y)$, the matrix element
$\bra{xy}\,H\,\ket{\bar{x}y}$ is a sum of at most $\poly(n)$
terms from those local terms $h_i$ whose support intersects $S$.
The map $a \mapsto |a|$ is efficiently computable, so individual diagonal matrix entries of $\widetilde{H}_{S,x}$ are computable in $\poly(n)$ time from the local description of $H$.

\item \emph{Same decomposition structure.}
Since the proof of Lemma~2 in Ref.~\cite{Ioannou_2020} constructs $\beta$, $U_j$, and $H_j$ entirely from the operators $H_{S,x}$---using only their diagonal structure, nonnegativity, and the ability to compute their entries efficiently---the same construction applied to $\widetilde{H}_{S,x}$ yields the decomposition~\eqref{eq:S_decomposition} for $\mathcal{S}(H)$.
\end{enumerate}
The total number of terms satisfies $m' \le m\, 2^{2k}$, and all components are efficiently computable, classically.
We stress that $\widetilde H_{S,x}$ need not decompose into a sum of $O(1)$-local diagonal terms---the absolute value of a sum of local functions is in general a nonlocal function---but the construction of Ref.~\cite{Ioannou_2020} uses only the properties (a)--(c) above, together with the norm bound $\|\widetilde H_{S,x}\| \leq \|H\| = \poly(n)$; none of these requires locality of the $\widetilde H_{S,x}$ themselves.
\end{proof}

Before assembling the reduction we isolate the one step at which we leave the setting of Ref.~\cite{Ioannou_2020}.
Their verifier accesses the Hamiltonian only by measuring the diagonal terms $H_j$ of the decomposition, and the measurement lemma they prove for this step assumes that each $H_j$ is classical \emph{and local}.
Our terms $\widetilde H_{S,x}$ are classical but need not be local, since an absolute value of a sum of local functions is in general a nonlocal function of $y$.
What their construction actually uses, however, is not locality itself but the ability to evaluate $H_j$ reversibly and efficiently with a polynomial-size classical circuit.
Their assumption of locality is only \emph{sufficient}, but not necessary, to guarantee that such a circuit exists for the input Hamiltonian.
The following generalizing lemma isolates that weaker property, so that the containment goes through for any efficiently evaluable nonnegative diagonal terms $H_j$ in the decomposition from~\Cref{lem:glob_stoq_decomp}, local or not.

\begin{lemma}[$\StoqMA$ measurement from an efficient, reversible Hamiltonian matrix entry oracle]
\label{lem:oracle-measurement}
Let $F$ be a nonnegative diagonal operator on $n$ qubits with $\|F\|\le M$.
Suppose there is a uniform, $\poly(n,b)$-size reversible circuit $\mathcal{O}_F$ that, on input $\ket{y}\ket{0}$, outputs $\ket{y}\ket{\widehat F(y)}$, where $\widehat F(y)$ is a $b$-bit approximation of $F(y)=\braket{y|F|y}$ with $|\widehat F(y)-F(y)|\le M\,2^{-b}$.
Then for every $\delta\ge 1/\poly(n)$ there is a valid $\StoqMA$ verifier---a circuit over $\{X,\mathrm{CNOT},\mathrm{Toffoli}\}$ acting on the witness together with $\poly(n)$ ancillas in $\ket{0}$ and $\ket{+}$, with a designated output qubit measured in the Hadamard basis---whose acceptance probability on any witness $\ket{\psi}$ is
\begin{align*}
p(\psi) = \alpha+ \beta\,\frac{\braket{\psi|X\otimes F|\psi}}{M} \pm \delta ,
\end{align*}
for fixed, efficiently computable constants $\alpha,\beta$ independent of $\psi$, taking
precision $b=\big\lceil\log_2(M/\delta)\big\rceil+O(1)$. The same holds with $X\otimes F$
replaced by $F$.
\end{lemma}

\begin{proof}
    See~\Cref{app:proof_of_oracle_measurement}.
\end{proof}
We emphasize that the above lemma is not necessary for the proof of~\Cref{thm:vgp_stoqma}.
Instead, it is a generalization of the key mechanism behind our proof (which is specific to VGP local Hamiltonians).
The insight that it provides potentially opens the door to proving $\StoqMA$ membership for other classes of Hamiltonians beyond VGP.

\subsubsection{Completing the Proof of \Cref{thm:vgp_stoqma}}

We now assemble the complete reduction.

\begin{proof}[Proof (of~\Cref{thm:vgp_stoqma})]
Let $H$ be a VGP local Hamiltonian.
The reduction proceeds in three steps.
\begin{enumerate}[label=(\roman*),leftmargin=2em]
    \item[$(i)$] Ioannou et al.~\cite{Ioannou_2020} prove that the local
    Hamiltonian problem for stoquastic local
    Hamiltonians is in $\StoqMA$
    (\Cref{thm:glob_stoq_stoqma}).
    Their proof requires only the decomposition from \Cref{lem:glob_stoq_decomp} and that the classical $H_j$ terms in the decomposition have efficiently computable classical energies.

    \item[$(ii)$] \Cref{prop:vgp_equiv_def} shows that $\mathcal{S}(H)$ is stoquastic and cospectral to $H$, when $H$ is promised to have VGP.

    \item[$(iii)$] Lemma~\ref{lem:stoq_map_decomp} shows that
    $\mathcal{S}(H)$ admits the required, efficiently computable decomposition and that the classical $H_j$ terms in the decomposition have efficiently computable energies
\end{enumerate}
Therefore the $\StoqMA$ verifier in Ref.~\cite{Ioannou_2020}, applied to the decomposition of $\mathcal{S}(H)$, estimates the ground-state energy of $H$.
\end{proof}

The mechanism behind the proof is now precise. The verifier of Ref.~\cite{Ioannou_2020} for a local stoquastic Hamiltonian accesses that Hamiltonian only through the decomposition of \Cref{lem:stoq_map_decomp} and the measurement of its diagonal terms $H_j$; by \Cref{lem:oracle-measurement}, each such measurement can be carried out to accuracy $\delta/m'$ given only a $\poly(n)$-size reversible oracle for the entries of $H_j$ together with the bound $\|H_j\| \leq \poly(n)$.
Applied to $\mathcal{S}(H)$: by \Cref{lem:stoq_map_decomp} its decomposition has $m' = \poly(n)$ terms; each entry $\braket{xy|\mathcal{S}(H)|\bar x y}= -\left|\braket{xy|H|\bar x y}\right|$ has magnitude equal to the absolute value of a sum of at most $\poly(n)$ local matrix elements of $H$, and so is computable by a uniform reversible circuit; and $\|\mathcal{S}(H)\|=\|H\|=\poly(n)$.
Choosing $\delta$ below the promise gap---which requires only $b=O(\log n)$ bits of precision, since $M = \poly(n)$ and $\delta = 1/\poly(n)$---the StoqMA verifier of Ref.~\cite{Ioannou_2020}, run on this oracle, decides the ground-state energy of $\mathcal{S}(H)$, and hence of $H$, since $\mathcal{S}(H)$ is cospectral to $H$ under the
VGP promise (\Cref{prop:vgp_equiv_def}).
Crucially, $\mathcal{S}(H)$ need not be local: the construction uses only the efficient reversible entry oracle and the $\poly(n)$ norm bound, both of which hold whether or not $\mathcal{S}(H)$ is local.

\subsection{Collapse: $\VGPMA = \StoqMA$}

\begin{corollary}
    \label{cor:vgpma_collapse}
    $\VGPMA = \StoqMA$.
\end{corollary}
\begin{proof}
    $\VGPLH \in \StoqMA$ by \Cref{thm:vgp_stoqma}.
    \Cref{prop:containments} gives $\StoqMA \subseteq \VGPMA$, yielding the collapse $\VGPMA = \StoqMA$.
\end{proof}

\begin{corollary}
\label{cor:conceptual}
Within PMR-QMC and Hamiltonian complexity, VGP characterizes the sign-problem-free boundary.
\end{corollary}

\begin{proof}
Theorem~\ref{thm:qmc_boundary} identifies VGP as the exact
sign-problem-free boundary for PMR-QMC.
\Cref{cor:stoqma_complete} gives $\StoqMA$-completeness of
$\VGPLH$, and Theorem~\ref{thm:vgp_stoqma} together with
Corollary~\ref{cor:vgpma_collapse} gives
$\VGPMA = \StoqMA$.
\end{proof}

The VGP promise asserts only the \emph{existence} of a stoquastizing  diagonal $D$, not that $D$ is
efficiently computable or has a succinct circuit description.
These are distinct properties.
The collapse $\VGPMA = \StoqMA$ shows that a stoquastizing  diagonal is
\emph{computationally irrelevant} for the ground-state energy decision problem---the VGP promise allows one to efficiently query entries of a cospectral stoquastic Hamiltonian without recovering $D$.
But a stoquastizing  diagonal may remain \emph{representationally inaccessible}: as established in \Cref{ex:n_qubit_family_non_local_stoquastizing _diagonal} and \Cref{thm:vgp_hard_to_stoquastize}, a stoquastizing  diagonal may map local Hamiltonians only to descriptions with exponentially many terms, admitting no efficiently realizable representation (\Cref{def:eff_realizable})---even when the stoquastizing diagonal itself has a simple circuit.

A natural follow-up question is: what is the complexity of recognizing VGP in $k$-local Hamiltonians?
We explore this in the next section in detail.
But first, we would like to classify the complexity of another VGP local Hamiltonian problem.

\subsection{A VGP Frustration-Free Local Hamiltonian Problem in $\MA$}

Suppose $H_0 = \sum_i D_iP_i$ is a $k$-PMR-local Hamiltonian that is promised to have VGP.
\Cref{cor:efficient_stoquastization_pmr_local} states that one may efficiently construct a decomposition of the cospectral stoquastic $\mathcal{S}(H_0)$ as a linear dependence relation in the Pauli operator basis.
Additionally, $\mathcal{S}(H_0)$ remains $k$-PMR-local.
Now decompose $\mathcal{S}(H_0) = \sum_i D'_iP'_i$ into PMR form.
For every PMR term, put $h_i := D'_iP'_i$ as a $k$-local Hamiltonian term.
Then $H := \sum_i h_i$ is a $k$-local, $k$-\emph{termwise stoquastic} Hamiltonian.
\begin{definition}[from Ref.~\cite{Ioannou_2020}]
    \label{def:termwise_stoquastic}
    A local Hamiltonian $H$ is considered $m$-\emph{termwise stoquastic} if $H$ admits a decomposition into $m$-local terms $H = \sum_i h_i$ such that each $h_i$ is Hermitian, real, and obeys $\bra{x}h_i\ket{y} \leq 0$ for all pairs $x \neq y$.
    \qed
\end{definition}
Additionally, by \Cref{prop:vgp_equiv_def}, $H_0$ is spectrally equivalent to $H$ under the VGP promise.
We use this convenient transformation as motivation for investigating whether the \emph{frustration-free local Hamiltonian problem} is in $\MA$ for all such $H_0$.
\begin{definition}
    \label{def:frustration_free_local_hamiltonian_prob}
    The \emph{frustration-free local Hamiltonian problem} is as follows:
    \begin{itemize}[leftmargin=5.5em]
        \item[\textbf{Input}:] A local Hamiltonian decomposed in PMR form $H=\sum_i h_i = \sum_i D_iP_i$, and a parameter $\varepsilon \geq 1/\poly(n)$.
        \item[\textbf{Promise}:] Either there exists a normalized state $\ket{\psi}$ such that $\bra{\psi}h_i\ket{\psi}=\lambda_{\min}(h_i)$ for every $i$, or for every normalized $\ket{\psi}$,
        \begin{align*}
            \sum_i \left(\bra{\psi}h_i\ket{\psi}-\lambda_{\min}(h_i)\right) \geq \varepsilon.
        \end{align*}
        \item[\textbf{Problem}:] Decide which is the case. \hspace{3cm} \qed
    \end{itemize}
\end{definition}
Ref.~\cite{Bravyi_ff_2009} proves that the frustration-free $k$-termwise stoquastic local Hamiltonian problem is contained in $\MA$ when the local terms in the decomposition are promised to be positive semidefinite (each term has all nonnegative eigenvalues).
They prove $\MA$-hardness for $k \geq 6$.
A common ``standardization'' for frustration-free Hamiltonians (see Ref.~\cite{Sattath_2016}) shows that the local terms need not be positive semidefinite for $\MA$-completeness to survive, so we do not need this extra assumption for the problem statement above.
We claim an analogous result holds for VGP Hamiltonians.
Because frustration-freeness is a property of the term decomposition rather than of the spectrum, we first record that a stoquastizing  diagonal acts termwise on a PMR decomposition:
\begin{lemma}[Termwise stoquastization]
    \label{lem:termwise_stoq}
    Let $H_0 = D_0 + \sum_i D_iP_i$ be cast in PMR form and suppose $D$ is a stoquastizing  diagonal for $H_0$.
    Then $D(D_iP_i)D^\dagger = \mathcal{S}(D_iP_i)$ for every $i$.
    Consequently, each term's spectrum is preserved, $\bra{\psi}D_iP_i\ket{\psi} = \bra{D\psi}\mathcal{S}(D_iP_i)\ket{D\psi}$ for every state $\ket{\psi}$, and $\ket{\psi} \mapsto D\ket{\psi}$ is a bijection between simultaneous minimizers of the terms $\{D_iP_i\}$ and of $\{\mathcal{S}(D_iP_i)\}$.
\end{lemma}
\begin{proof}
    PMR terms are matrix entry-wise disjoint, so every non-zero off-diagonal entry of $H_0$ belongs to exactly one term.
    The stoquastization constraint $e^{i(\theta_x - \theta_y)}(H_0)_{xy} = -|(H_0)_{xy}|$ therefore acts on each term separately, giving $D(D_iP_i)D^\dagger = \mathcal{S}(D_iP_i)$; the remaining claims follow from unitarity.
\end{proof}

\begin{theorem}
    \label{thm:ff-pmr-local-lh-ma}
    The frustration-free local Hamiltonian problem is in $\MA$ for PMR-local Hamiltonians $H_0$ that are promised to have VGP.
\end{theorem}
\begin{proof}
    Mapping $H_0$ to the $k$-termwise stoquastic Hamiltonian $\mathcal{S}(H_0) = \sum_i \mathcal{S}(D_iP_i)$ can be done efficiently, as we argued above.
    By \Cref{lem:termwise_stoq}, $\ket{\psi} \mapsto D\ket{\psi}$ matches simultaneous minimizers of the two term sets and preserves the total violation $\sum_i (\bra{\psi}h_i\ket{\psi} - \lambda_{\min}(h_i))$, so $H_0$ is a YES (resp.\ NO) instance if and only if $\mathcal{S}(H_0)$ is, with the same promise gap.
    Thus the decision problem is in $\MA$ by the proof in Ref.~\cite{Bravyi_ff_2009}.
\end{proof}
We remark that $\MA$-completeness of the above decision problem remains open; the $\MA$-hardness construction of \cite{Bravyi_ff_2009} concerns whether a normalized state can simultaneously minimize the local terms of a constructed clock Hamiltonian, but not necessarily its PMR terms.
So, their proof of $\MA$-hardness cannot immediately be ported here.

This gives a frustration-free local Hamiltonian problem in $\MA$ for a natural subclass of VGP Hamiltonians that need not be stoquastic in the input basis, and for which constructing a stoquastizing diagonal may be computationally intractable: discussion of \Cref{thm:loc_vgp_hardness} demonstrates that no efficiently realizable stoquastizing diagonal may exist in the PMR-local case, and we will prove in the next section how difficult it can be to decide the existence of a stoquastizing diagonal for PMR-local $H$.

\section{Complexity of VGP Recognition}
\label{sec:vgp_complexity}

Since VGP characterizes the structural boundaries around $\StoqMA$, we pose the next follow-up question: how hard is it to recognize where a given local Hamiltonian problem sits relative to these boundaries?
That is, what is the complexity of recognizing VGP in a local Hamiltonian?

It may be helpful to review some preliminary definitions of some relevant complexity classes before we proceed. For these definitions, see Appendix~\ref{app:complexity_class_definitions}.

Now, we will proceed with our discussion of the complexity of recognizing VGP.
We will compare to the following problem:
\begin{definition}The \emph{stoquasticity recognition problem} is:
    \begin{itemize}[leftmargin=5.5em]
        \item[\textbf{Input}:] $H \gets$ $k$-local Hamiltonian on $n$ qubits
        \item[\textbf{Problem}:] Decide if $H$ is stoquastic.\qed
    \end{itemize}
\end{definition}
The stoquasticity recognition problem for $k$-local Hamiltonians is $\coNP$-complete when $k > 2$~\cite{Ioannou_2020} and can be done efficiently when $k = 2$~\cite{Klassen_2019}. 
We study the same problem for VGP:
\begin{definition}
    The \emph{VGP recognition problem} is:
    \begin{itemize}[leftmargin=5.5em]
        \item[\textbf{Input}:] $H \gets$ $k$-local Hamiltonian on $n$ qubits
        \item[\textbf{Problem}:] Decide if $H$ has VGP.\qed    
    \end{itemize}
\end{definition}
We find that VGP recognition for 2-local Hamiltonians is $\coNP$-hard~(\Cref{thm:2loc_vgp_hardness}), and conjecture that it is $\coNP$-complete.
For local Hamiltonians and geometrically local Hamiltonians in general, we find that VGP recognition is $\PSPACE$-complete~(\Cref{thm:loc_vgp_hardness}).
Nonetheless, in the next section, we will identify a broad, natural class of Hamiltonians for which the VGP property can be efficiently certified.
What's more, in these cases, a locality-preserving stoquastizing diagonal can also be obtained in polynomial time (\Cref{thm:dstoq_efficient},\Cref{cor:dstoq_ell}).

\subsection{Recognizing VGP in 2-local Hamiltonians}
\begin{theorem}
    \label{thm:2loc_vgp_hardness}
    The VGP recognition problem is $\coNP$-hard for 2-local Hamiltonians.
\end{theorem}
\begin{proof}
    See Appendix~\ref{app:proof_of_twolocalvgp_conphard}.
\end{proof}

$\coNP$ membership of the VGP recognition problem for 2-local Hamiltonians is currently unknown. 
One natural approach to proving $\coNP$ membership is as follows:
\begin{enumerate}
    \item For any $\ket{x}$ and $\ket{y}$ in the same connected component $C$ of $G_H$, there is a walk from $\ket{x}$ to $\ket{y}$ of size $O(\poly(n))$. This proves that the \emph{diameter}, or the longest shortest-path distance between any two vertices, of $C$ is polynomial in size.
    \item We can always choose a \emph{cycle basis} of an undirected graph $G$ that only consists of \emph{isometric cycles}\footnote{Meaning the distance between $u$ and $v$ in the cycle is the same as the shortest distance in the ambient graph}.
    The size of an isometric cycle of $C$ is at most linear in the diameter of $C$. 
    \item $H$ does not have VGP if and only if for every choice of cycle basis $\mathcal{C}$ for $G_H$, there exists some $\gamma \in \mathcal{C}$ with non-trivial holonomy (\Cref{prop:vgp_equiv_def}). If we choose $\mathcal{C}$ to only consist of isometric cycles, such a $\gamma$ would be a polynomial-size witness that certifies $H$ does not have VGP.
    This would place VGP recognition in $\coNP$.
\end{enumerate}
Once Step 1 is shown, the rest of the proof follows. Step 1 is currently a conjecture (Problem \ref{prob:2_local_poly_diam}).

If the VGP recognition problem is indeed $\coNP$-complete for 2-local Hamiltonians, then the problem of VGP recognition is efficiently reducible to a Boolean formula $\psi$ of an UNSAT instance (by the Cook-Levin theorem~\cite{Cook_1971,Levin_1973}).
Such a reduction could hypothetically allow one to use a modern SAT solver to decide if a given Hamiltonian has VGP, often in near-linear time (in practice)~\cite{Biere_2021}.
However, this is not guaranteed, especially if the reduction is not very straightforward.

\subsection{Recognizing VGP in $k$-local \& Geometrically Local Hamiltonians}
Recognition of VGP experiences a sharp transition in difficulty once 5-local interactions are allowed; even in a geometrically local setting.
\begin{theorem}
    \label{thm:loc_vgp_hardness}
    The VGP recognition problem is $\PSPACE$-complete for local Hamiltonians, geometrically local, and PMR-local Hamiltonians.
\end{theorem}

\begin{proof}
    We prove $\PSPACE$-hardness for real 5-local Hamiltonians that are also geometrically local.
    See Appendix~\ref{app:proof_of_threelocalvgp_pspacecomplete}.
\end{proof}
\begin{remark}
    Recall that the local Hamiltonian problem restricted to VGP Hamiltonians is $\StoqMA$-complete (\Cref{cor:stoqma_complete}).
    The above result shows that determining membership in the VGP class is intractable---even harder than the local Hamiltonian problem itself (which is $\QMA$-complete, and $\QMA \subseteq \PSPACE$).
    \qed
\end{remark}

By \Cref{prop:vgp_equiv_def}, a Hamiltonian $H$ has VGP if and only if $H$ has a stoquastizing  diagonal.
Thus we get an equivalent statement as a corollary:
\begin{corollary}
    \label{cor:pspace_stoquastizing_diagonal}
    Deciding if there exists a diagonal unitary $D$ such that $DHD^\dagger$ is stoquastic is $\PSPACE$-complete for local Hamiltonians, geometrically local, and PMR-local Hamiltonians $H$. \qed
\end{corollary}

\begin{table*}[t]
\label{tab:vgp_vs_stoq_complexity}
\begin{tabular}{c|c|c|c|c}
 Decision problem & 2-local& $k$-local, $k = O(1)$ & geometrically local & PMR-local 
\\\hline
 VGP recognition& $\coNP$-hard & $\PSPACE$-complete& $\PSPACE$-complete& $\PSPACE$-complete\\
 $\DSTOQ$& $\P$ & $\in \Sigma_2^p$ & $\P$ & $\P$ \\\hline
 Stoquasticity recognition& $\P$~\cite{Klassen_2019}& $\coNP$-complete~\cite{Ioannou_2020}& $\P$& $\P$\\
 $\USTOQ$& $\NP$-hard~\cite{Klassen_2020} & $\Sigma_2^p$-complete~\cite{Ioannou_2020}& $\in \Sigma_2^p$~\cite{Ioannou_2020} & $\in \Sigma_2^p$~\cite{Ioannou_2020} \\
\end{tabular}
\caption{Computational complexity of decision problems related to VGP versus stoquasticity. Entries in the table without a citation are proven in this paper. Entries marked ``$\in$'' are membership upper bounds without matching hardness results; in particular, there are no known hardness results for $\DSTOQ$ in the $k$-local case or for $\USTOQ$ restricted to geometrically local or PMR-local Hamiltonians.}
\end{table*}

This result contrasts quite strikingly with our findings in~\Cref{sec:efficient_stoquastization}.
PMR-local Hamiltonians $H$ that are promised to have VGP can be efficiently transformed into a local, cospectral Hamiltonian (\Cref{cor:efficient_stoquastization_pmr_local}), but the diagonal $D_H$ that sends $H \mapsto \mathcal{S}(H)$ need not be efficiently realizable.
We will argue this.
Assume that---in full generality---$D_H$ can always be efficiently realized when $H$ is $k$-PMR-local.
To verify that $D_H$ indeed stoquastizes $H$, one must check
\begin{enumerate}
    \item[$(i)$] that $D_H$ preserves $k$-PMR locality of $H$ (see proof of~\Cref{prop:efficient_stoquastization_pmr_local} for an explanation)
    \item[$(ii)$] and that $(D_H)H(D_H)^\dagger$ is stoquastic.
\end{enumerate}
$(i)$ may be checked efficiently, and if $(i)$ holds then $(ii)$ may also be checked efficiently (see \Cref{prop:efficient_stoq_check_pmr_local}).
Hence the VGP recognition problem for PMR-local Hamiltonians \emph{would be} in $\NP$, leading to the unlikely collapse $\PSPACE = \NP$~\cite{Arora_2009}.
Thus---assuming $\PSPACE \neq \NP$---we have demonstrated that there exist VGP local Hamiltonians $H$ for which the stoquastic proxy $\mathcal{S}$(H) is also local, but there is no efficiently realizable unitary $D_H$ sending $H \mapsto \mathcal{S}(H)$.

In the next section, we will investigate how the VGP recognition problem relates to recovering a stoquastizing  diagonal.
We will show that it is efficient to decide the existence of and find a local block-diagonal stoquastizing diagonal for $H$ under certain natural conditions (\Cref{thm:dstoq_efficient}).

\section{Complexity of Recovering Stoquastizing  Diagonal Unitaries}
\label{sec:diagonal_recovery}

The collapse $\VGPMA = \StoqMA$ shows that a stoquastizing  diagonal is unnecessary for ground-state energy estimation of VGP Hamiltonians.
These results clarify why VGP is the right structural language even though explicit stoquastization is difficult.
For QMC simulation, however, the diagonal is not always dispensable.
A sign-problem-free simulation of a VGP Hamiltonian $H$ is, in effect, a simulation of its stoquastic representative $\mathcal{S}(H) = DHD^\dagger$: expectation values are basis-independent, but QMC \emph{estimators} are not---they are built from matrix elements in the simulated representation, and a physical observable $O$ of $H$ corresponds in that representation to $DOD^\dagger$.
For diagonal observables the distinction is invisible, since $[D, O] = 0$; this is the diagonal-unitary-invariant sector identified in \Cref{sec:operational}.
An off-diagonal observable, by contrast, acquires the relative phases of the diagonal unitary on precisely the matrix elements where it is supported, $(DOD^\dagger)_{xy} = e^{i(\theta_x - \theta_y)}\,O_{xy}$, so the same observable may admit a directly constructible estimator in one representation and not in another.
When the flip pattern of $O$ coincides with an edge of $G_H$---as it does for the off-diagonal terms of $H$ itself, and hence for kinetic-energy-type estimators~\cite{Gupta_2020, Babakhani_2026}---the phase is locally computable from the Hamiltonian, since stoquastization forces $e^{i(\theta_x-\theta_y)} = -|H_{xy}|/H_{xy}$ on every edge.
But when $O$ connects basis states only through long walks in $G_H$, the required phase is a product of edge phases along a connecting path---well-defined and path-independent under the VGP promise, yet potentially expensive to obtain---and knowledge of an explicit $D$ obviates the search altogether: with $D$ in hand, $DOD^\dagger$ is constructed directly, and it remains local whenever $O$ is local and $D$ is block-diagonal.
This motivates our next question: how hard is it to recover a stoquastizing  diagonal unitary explicitly? But first, it is helpful to review established results for a related problem in the literature surrounding stoquasticity:

\begin{definition}[$\USTOQ$]
    For every $n$,
    fix a partition $\{B_\alpha\}$ of $[n]$ so that every $|B_\alpha| = O(1)$.
    Let
    \begin{align*}
        \mathcal{U}_n = \{ U : U = \bigotimes_\alpha U_{\alpha}, \,\, \supp(U_{\alpha}) = B_\alpha \}
    \end{align*}
    be a family of unitaries.
    \begin{itemize}[leftmargin=5.5em]
        \item[\textbf{Input}:] $H \gets$ $k$-local Hamiltonian on $n$ qubits

        \item[\textbf{Promise}:] In a YES instance, the witness $U$ is described in such a way that allows one to efficiently compute matrix entries of $UHU^\dagger$.
        \item[\textbf{Problem}:] Decide if there exists $U \in \mathcal{U}_n = \mathcal{U}$ such that $UHU^\dagger$ is stoquastic. \qed
    \end{itemize}
\end{definition}
Ref.~\cite{Ioannou_2020} proves that $\USTOQ$ is $\Sigma_2^p$-complete.
This holds even when $\mathcal{U}$ consists of single-qubit unitaries and $H$ is 3-local~\cite{Ioannou_2020}.
$\USTOQ$ is also $\NP$-hard for 2-local Hamiltonians~\cite{Klassen_2020}.

For this section, we examine the same problem, but restricted to diagonal unitaries:
\begin{definition}
    Define $\DSTOQ$ to be $\USTOQ$ with the promise that $\mathcal{D} = \mathcal{U}$ only consists of diagonal unitaries. \qed
\end{definition}

When $k = O(1)$, we immediately obtain the upper bound $\DSTOQ \in \Sigma_2^p$.
We argue $\Sigma_2^p$ membership as follows.
Recall that deciding if a Hamiltonian $H$ with efficiently computable matrix entries is stoquastic is in $\coNP$~\cite{Ioannou_2020}.
Hence if $D \in \mathcal{D}$ is provided as a (polynomial-sized) witness, one may efficiently compute $DHD^\dagger$ and query a $\coNP$ oracle to determine if $DHD^\dagger$ is stoquastic.
Thus $\DSTOQ \in \NP^{\coNP} = \Sigma_2^p$.

If we allowed $\mathcal{D}$ to be the set of all diagonal unitaries, then the problem reduces to the VGP recognition problem---the existence of any stoquastizing  diagonal implies that $H$ has VGP (\Cref{prop:vgp_equiv_def}).
The block-diagonal restriction makes the problem weaker and does not guarantee non-VGP certification in every NO case, but solving the problem can be significantly easier.
Of course, every YES instance of $\DSTOQ$ certifies that $H$ has VGP.

\begin{theorem}
    \label{thm:dstoq_efficient}
    $\DSTOQ$ is decidable in polynomial time when $H$ is at least one of the following:
    \begin{itemize}
        \item[(i)] 2-local
        \item[(ii)] geometrically local
        \item[(iii)] PMR-local
    \end{itemize}
    Additionally, in every YES instance one can recover a stoquastizing  $D \in \mathcal{D}$ for $H$ in polynomial time.
\end{theorem}
\begin{proof}
    Fix an instance of $\DSTOQ$ with verifier $V$.
    We will create a system of constraints for the 2-local case and the PMR-local case such that the system has a solution if and only if some $D \in \mathcal{D}$ stoquastizes $H$.
    We will then show that these constraints can be converted into polynomially-many linear equations on $O(1)$ variables.

    We prove necessary and sufficient conditions in the 2-local case first.
    Assume $H = \sum_{i=1}^m D_iP_i$ is 2-local and cast in PMR form.
    We will construct a system of linear equations that has a solution if and only if $V(H)$ outputs YES.
    The system will be solvable in polynomial time, and each PMR term will contribute $O(1)$ equations and variables.
    
    For every $i \in [m]$, we will do a case-by-case analysis for $D_iP_i$:
    \begin{enumerate}
        \item Assume $|\supp(P_i)| = \{p\}$.
        Then we may write $P_i = X_p$.
        We know that $H_{xy} =\bra{x}D_iP_i\ket{y}$ can only be non-zero when $x$ and $y$ differ on qubit $p$.
        Fix the two states $s,t \in \{0,1\}^{|B_{\alpha}|}$, the restrictions of $x$ and $y$ to the block $B_{\alpha}$ containing $p$ (i.e., $p \in B_{\alpha}$); since $x$ and $y$ differ only at qubit $p$, $t$ equals $s$ with bit $p$ flipped.
        Let $z$ be the restriction of $x$ (or $y$) to the qubits outside of $B_{\alpha}$.
        Since $H$ is 2-local, we can decompose
        \begin{align*}
            H_{xy} = F_{\alpha,s,t}(z) := a + \sum_{j \notin B_\alpha} b_j \sigma_j, && \sigma_j := 2z_j - 1.
        \end{align*}
        With $s$ and $t$ fixed, the only constant is $a$ (which absorbed the contributions from the assignments of $s$ and $t$).
        Notice that when $D \in \mathcal{D}$ conjugates $H$,
        \begin{align*}
            H_{xy} \mapsto e^{i\phi_{s,t}}H_{xy}
        \end{align*}
        a mapping independent of $z$ since $x$ and $y$ agree on $z$.
        Letting $z$ vary, $D$ being a stoquastizing  diagonal for $H$ can only be possible if
        \begin{align*}
            e^{i\phi_{s,t}}F_{\alpha,s,t}(z) \in \mathbb{R}_{\leq 0} && z \in \{0,1\}^{n - |B_{\alpha}|}.
        \end{align*}
        Consequently, for every choice of $z$ with $F_{\alpha,s,t}(z) \neq 0$, $F_{\alpha,s,t}(z)$ has the same complex phase (namely, $\pi - \phi_{s,t}$ modulo $2\pi$).
        In other words, there is a common complex number $v$ so that for every $z$,
        \begin{align*}
            F_{\alpha,s,t}(z) = v\underbrace{\left( A + \sum_{j} B_j\sigma_j \right)}_{=:R(z)}
        \end{align*}
        with $a = vA$, $b_j = vB_j$, and $A,B_j$ real.
        This requirement is equivalent to $a$ and every coefficient $b_j$ being a real multiple of the same complex number.
        For each choice of $s$ and $t$, one can check whether this is the case efficiently for all $O(n)$ coefficients and recover $v$.
        
        Now we focus on the real part $R(z)$.
        The only way for $D$ to stoquastize $H$ is if, when $R$ is minimized, it is nonnegative, or when it is maximized it is non-positive.
        Otherwise, $R(\,\cdot\,)$ is positive on some inputs and negative on others---no $D \in \mathcal{D}$ can stoquastize $H$ if that is the case.
        Equivalently, the requirement is
        \begin{align*}
            A - \sum_j |B_j| \geq 0 && \text{or} && A + \sum_j |B_j| \leq 0.
        \end{align*}
        In the first case, stoquasticity imposes the constraint $e^{i\phi_{s,t}} = -\frac{\overline{v}}{|v|}$.
        In the second case, it imposes the constraint
        $e^{i\phi_{s,t}} = \frac{\overline{v}}{|v|}$.
        Iterating over all pairs $s$ and $t$ contributes $O(1)$ constraints.
        Each constraint may be expressed as linear equations on $O(1)$ variables (complex phases) modulo $2\pi$ since each $\phi_{s,t}$ depends only on $O(1)$ qubits.
        \item Assume $\supp(P_i) = \{p,q\}$.
        Stoquastization of $D_iP_i$ under some $D \in \mathcal{D}$ requires  $\bra{x}D_iP_i\ket{y} = H_{xy} \mapsto -|H_{xy}|$.
        (Since $H$ is $2$-local, $D_i$ is supported on $\{p,q\}$ as well, so these entries do not depend on spectator qubits; the linear constraints for this case are derived jointly with the PMR-local case below.)
    \end{enumerate}

    The above cases are the two possibilities for a PMR term of $H$ when $H$ is 2-local.
    For the case that $H$ is PMR-local, we have the same constraint: a YES instance of $\DSTOQ$ requires $H_{xy} \mapsto -|H_{xy}|$ for every $x,y \in \{0,1\}^n$ under some $D \in \mathcal{D}$.

    We argue that every such     constraint may be reduced to polynomially-many linear equations modulo $2\pi$.
    In case 1 of the 2-local case, this is already done.
    So, we fixate on the PMR-local case and case 2 of the 2-local case.
    Write 
    \begin{align*}
        D = \diag(e^{i\theta_z}) = \bigotimes_\alpha D_\alpha = \prod_\alpha \diag( \{e^{i\theta^{(\alpha)}_z} \})
    \end{align*}
    as block-diagonal over $\{B_\alpha\}$ so that
    \begin{align*}
        \theta_z \equiv \sum_{\alpha} \theta_z^{(\alpha)} \quad \text{(mod $2\pi$)}.
    \end{align*}
    Then $D$ being a stoquastizing  diagonal requires
    \begin{align}
        (DHD^\dagger)_{xy} &= e^{i\theta_{x} - \theta_{y}}H_{xy} = -|H_{xy}| && \implies\\
        \label{eq:entry_stoq_constraint}
        \theta_{x} - \theta_{y} &\equiv \pi - \arg(H_{xy})  \quad \text{(mod $2\pi$)}.
    \end{align}
    For $H_{xy} \neq 0$, $x$ and $y$ must differ on at most $O(1)$ qubits $S$ by the locality of $H$ (there exists some PMR term $D_iP_i$ so that $\bra{x}D_iP_i\ket{y}\neq 0$).
    Let $\{P_\alpha\} \subseteq \{B_\alpha\}$ be the subset of blocks overlapping with $S$.
    Now examine the difference
    \begin{align*}
        \theta_x - \theta_y = \left[\sum_{P_\alpha \in \{B_\alpha\}} \theta_x^{(\alpha)} - \theta_y^{(\alpha)} \right] + \left[\sum_{P_\alpha \notin \{B_\alpha\}} \theta_x^{(\alpha)} - \theta_y^{(\alpha)}\right]. 
    \end{align*}
    For $P_\alpha \notin \{B_\alpha\}$, we have $\theta_x^{(\alpha)} = \theta_y^{(\alpha)}$ since $x$ and $y$ agree on qubits outside of $\{P_\alpha\}$.
    Hence the constraint in~\Cref{eq:entry_stoq_constraint} reduces to
    \begin{align}
        \label{eq:phase_requirement}
        \sum_{P_\alpha \in \{B_\alpha\}} \theta_x^{(\alpha)} - \theta_y^{(\alpha)} \equiv \pi - \arg(H_{xy}) \quad (\text{mod }2\pi).
    \end{align}
    In the PMR-local case, $H_{xy}$ depends on $O(1)$ qubits by definition of PMR locality.
    In case 2 of the 2-local case, $H_{xy}$ depends on 2 qubits.
    So, $|\{P_\alpha\}| = O(1)$, admitting at most $O(1)$ variables on the LHS of the linear constraint in~\Cref{eq:phase_requirement}.
    Letting $x$ and $y$ vary on the qubits of $\supp(D_iP_i) \cup \left( \bigcup_\alpha P_\alpha \right)$ means that \Cref{eq:phase_requirement} contributes $O(1)$ linear equations, each depending on $O(1)$ variables.
    Since there are $\poly(n)$ PMR terms, there are $\poly(n)$ such linear equations, so the system is solvable in $\poly(n)$ time. 
    The system of constraints has a solution if and only if a stoquastizing  diagonal $D \in \mathcal{D}$ exists for $H$.
    This completes the proof.
\end{proof}

We can generalize the problem so that one may choose---in advance---at most $O(1)$ overlapping local support blocks and efficiently decide if there exists a stoquastizing diagonal with support on those blocks:
\begin{definition}[$\DSTOQl$]
    For every $n$,
    fix $\ell$ partitions $\{B_\alpha^{(1)}\}, \ldots, \{B_\alpha^{(\ell)}\}$ of $[n]$ so that every $|B_\alpha^{(i)}| = O(1)$.
    Let
    \begin{align*}
        \mathcal{D}_n = \{ D : D = \prod_{i=1}^\ell \bigotimes_\alpha D_{\alpha}^{(i)}, \,\, \supp(D_{\alpha}^{(i)}) = B_\alpha^{(i)} \}
    \end{align*}
    be a family of diagonal unitaries.
    \begin{itemize}[leftmargin=5.5em]
        \item[\textbf{Input}:] $H \gets$ $k$-local Hamiltonian on $n$ qubits

        \item[\textbf{Promise}:] In a YES instance, the witness $D$ is efficiently realizable.
        \item[\textbf{Problem}:] Decide if there exists $D \in \mathcal{D}_n = \mathcal{D}$ such that $DHD^\dagger$ is stoquastic. \qed
    \end{itemize}
\end{definition}

\begin{corollary}
    \label{cor:dstoq_ell}
    $\DSTOQl$ is decidable in polynomial time on Hamiltonian inputs $H$ that are PMR-local (or geometrically local) when $\ell = O(1)$.
    Furthermore, in every YES instance, one can efficiently recover a stoquastizing  diagonal $D \in \mathcal{D}$ for $H$.
\end{corollary}
\begin{proof}
    Conjugation by any $D \in \mathcal{D}$ preserves the PMR locality of $H$, since each $|B_\alpha^{(i)}| = O(1)$; in particular, every entry $H_{xy} \neq 0$ still depends on $O(1)$ qubits (\Cref{rem:pmr_local_poly_many_entries}).
    Write $\theta_z \equiv \sum_{i=1}^{\ell} \theta_z^{(i)} \pmod{2\pi}$ with each layer block-diagonal over its own partition.
    The stoquastization constraints of~\Cref{eq:phase_requirement} then become linear equations modulo $2\pi$ in the union of all layers' block phase variables, to be solved simultaneously: for each non-zero entry $H_{xy}$, at most $O(\ell)$ blocks across the $\ell$ partitions overlap the $O(1)$ qubits on which the entry depends, so with $\ell = O(1)$ each PMR term contributes $2^{O(\ell)} = O(1)$ equations in $O(1)$ variables.
    The resulting system of $\poly(n)$ linear equations is solvable in polynomial time, and it has a solution if and only if some $D \in \mathcal{D}$ stoquastizes $H$; any solution yields the stoquastizing  $D$.
\end{proof}

\begin{remark}
    The key difference between $\DSTOQ$ and deciding if \emph{any} stoquastizing  diagonal $D'$ exists for PMR-local $H$ is that $D'$ does not necessarily have to be block-diagonal.
    
    The constraint is $H_{xy} \mapsto -|H_{xy}|$ for every pair of computational basis states $\ket{x}$ and $\ket{y}$.
    For the existence of an arbitrary stoquastizing  diagonal, this may create $O(2^n)$ new constraints per $H_{xy} \neq 0$. 
    However, when one is only considering local block-diagonal stoquastizing  diagonals $D$, $H_{xy}$ is transformed by phases with support on $O(1)$ qubits.
    What's more, if $H$ is PMR-local, every entry $H_{xy}$ depends on $O(1)$ qubits.
    These are precisely the two properties that allow one to reduce the $O(2^n)$ constraints to a tractable number of linear equations.
    One is not afforded such properties when searching arbitrary diagonals.
    \qed
\end{remark}

Remarkably, if one restricts each $D \in \mathcal{D}$ to a given family, the decision problem can become $\NP$-complete:
\begin{theorem}\label{thm:diag_stoq_hardness}
    The following problem is $\NP$-complete:
    
    Given a PMR-local Hamiltonian $H$ and a family $\mathcal{D}$ of locality-preserving diagonal unitaries as input, decide if there exists $D \in \mathcal{D}$ such that $DHD^\dagger$ is stoquastic.
\end{theorem}
\begin{proof}
    See Appendix~\ref{app:proof_of_nphardness_pmrlocal_vgprecognition_promiseproblem}.
\end{proof}
\Cref{thm:diag_stoq_hardness} demonstrates that hardness of stoquastization can arise from restricting the set of allowed unitary transformations.
Intuitively this should not make the search more difficult, because it shrinks the search space after all.
But the search can nonetheless become intractable.
This contrasts heavily with the fact that searching for an unrestricted single-qubit stoquastizing diagonal can be done efficiently (\Cref{thm:dstoq_efficient}).

Thus we have demonstrated that VGP recognition can be significantly easier than stoquastizing  a 2-local or geometrically local Hamiltonian with arbitrary local unitary transformations (this is in reference to the complexity of $\USTOQ$).
Consequently, attempting to stoquastize a local Hamiltonian with arbitrary unitaries may be unnecessary in certain cases.
Under natural conditions, one can efficiently screen the Hamiltonian to determine if it has VGP and if a convenient choice of locality-preserving stoquastizing diagonal exists.
In a YES instance, this would establish that the given Hamiltonian has VGP while recovering a stoquastizing  diagonal---allowing one to study the dynamics of the VGP Hamiltonian and its stoquastic proxy via sign-problem-free PMR-QMC simulation (\Cref{thm:qmc_boundary}).

\section{Operational Distinctions Between VGP and Stoquastic
  Hamiltonians}
\label{sec:operational}

The collapse $\VGPMA = \StoqMA$ establishes that VGP and stoquastic local Hamiltonians are indistinguishable for ground-state energy estimation.
This section identifies which computational tasks are insensitive
to the VGP--stoquastic distinction, and which appear to genuinely
separate the two classes.

\subsection{Tasks Not Separating VGP From Stoquasticity}
\label{subsec:no_separation}

\subsubsection{Ground-state energy}
This is the content of \Cref{thm:vgp_stoqma}.

\subsubsection{Partition function and thermodynamics}
The partition function $Z = \mathrm{Tr}(e^{-\beta H})$ is a
spectral quantity: under the VGP promise, $Z(H) = Z(\mathcal{S}(H))$, and all equilibrium thermodynamic observables derived from it (free energy, entropy, specific heat) are identical for $H$ and its stoquastic proxy.
Moreover, the closed-walk weights in the PMR expansion of $Z$ depend only on the diagonal-unitary-invariant data $|H_{xy}|$ and $H_{zz}$, so under the VGP promise the expansions of $H$ and $\mathcal{S}(H)$ agree term by term. (We emphasize that sign-problem-freeness of this expansion does not by itself imply efficient sampling.
See the final open problem of \Cref{sec:discussion}).

\subsubsection{Diagonal observables}
For any observable $O$ that is diagonal in the computational
basis,
$\braket{\psi_0|O|\psi_0} = \sum_z a_z^2\, O_{zz}$
depends only on the magnitudes
$a_z = |\braket{z|\psi_0}|$, which by \Cref{thm:vgp_pf} coincide
with the ground-state amplitudes of $\mathcal{S}(H)$ on the
relevant component.
Every diagonal equilibrium expectation value is therefore identical for $H$ and its stoquastic proxy, and is accessible to any method that samples from $a_z^2$.

For general ($k$-local but off-diagonal) observables the situation is more delicate.
The expectation value
$\braket{\psi_0|O|\psi_0}
 = \sum_{x,y} e^{i(\theta_x - \theta_y)}\, a_x\, O_{xy}\, a_y$
involves the relative phases $\theta_x - \theta_y$.
Under the VGP promise, these are diagonal-unitary-invariant along any path of $G_H$ connecting $y$ to $x$, but the pairs $(x,y)$ selected by $O$ need not be connected by short paths in $G_H$---or connected at all---and a sampler for $a_z^2$ does not by itself supply the required amplitude data.
We therefore leave the accessibility of off-diagonal local observables from stoquastic-proxy data as an open question; it is one of the motivations for the results of
\Cref{sec:diagonal_recovery} concerning the complexity of recovering a stoquastizing diagonal.

\subsection{Tasks Separating VGP From Stoquasticity}
\label{subsec:separation}

The VGP--stoquastic distinction becomes operationally meaningful
for tasks that require the global phase structure of the
ground-state rather than spectral or diagonal data.

\subsubsection{Classical ground-state description}
For a stoquastic Hamiltonian, the Perron--Frobenius ground-state has $\braket{z|\psi_0} \geq 0$ on each connected component and is completely characterized by the probability distribution $p(z) = |\braket{z|\psi_0}|^2$---a purely classical object.
For a VGP Hamiltonian, the ground-state has the form $\braket{z|\psi_0} = e^{-i\theta_z} a_z$: the magnitudes are available from the stoquastic proxy, but a complete description also requires the phases $\{\theta_z\}$.
By \Cref{ex:n_qubit_family_non_local_stoquastizing _diagonal} and \Cref{thm:vgp_hard_to_stoquastize}, there are VGP Hamiltonians for which no stoquastizing diagonal is efficiently representable in the sense of \Cref{def:eff_realizable}.
We note that this is a statement about \emph{conjugation}---no diagonal maps the local description of $H$ to a polynomial-size local description of a stoquastic Hamiltonian---and not about the description length of the phase function itself, which can be very simple ($\theta_z = \pi$ exactly on the all-ones string in \Cref{ex:n_qubit_family_non_local_stoquastizing _diagonal}).
\subsubsection{Quantum state preparation}
The ground state $\ket{\phi_0}$ of a stoquastic Hamiltonian is nonnegative in the computational basis, so it is fixed entirely by its amplitude data $\{a_z\}$; the VGP ground state carries the same amplitude data rotated by the stoquastizing diagonal, $\ket{\psi_0}=D^{\dagger}\ket{\phi_0}$ with $\braket{z|\psi_0}=e^{-i\theta_z}a_z$ (\Cref{thm:vgp_pf}).
Any route that produces $\ket{\phi_0}$ from its amplitudes therefore yields, for a VGP Hamiltonian, only the proxy state $\ket{\phi_0}$; recovering the physical ground state $\ket{\psi_0}$ requires the additional step of applying $D^{\dagger}$.
When a circuit for $D$ is known---as in \Cref{ex:n_qubit_family_non_local_stoquastizing _diagonal}, where $D$ is a single multi-controlled-$Z$---this step is available; what \Cref{ex:n_qubit_family_non_local_stoquastizing _diagonal} and \Cref{thm:vgp_hard_to_stoquastize} rule out is a stoquastizing diagonal that is efficiently realizable in the sense of \Cref{def:eff_realizable}, i.e., one admitting a compact representation of its conjugation action; no general procedure for \emph{finding} a circuit for $D$ is known.

This is a statement about stoquastizing diagonal step, not a claim that stoquastic ground states are classically easy to prepare.
Non-negativity of $\{a_z\}$ removes the sign obstruction to coherent state-preparation primitives---the target amplitudes $a_z = \sqrt{p(z)}$ are real and nonnegative, so no relative phases need be synthesized---but those primitives still require coherent access to the amplitudes and are efficient only under additional structure.
Classical sampling from $p(z) = a_z^2$, which sign-free PMR-QMC provides when it mixes, yields the mixture $\sum_z a_z^2\,\ket{z}\!\bra{z}$ rather than the coherent $\ket{\phi_0}=\sum_z a_z\ket{z}$, and amplitude amplification needs a coherent reflection oracle, not samples; neither prepares a ground state by itself.
What separates the two classes is that, whatever the cost of preparing $\ket{\phi_0}$, obtaining $\ket{\psi_0}$ from stoquastic-proxy data carries the additional cost of recovering a stoquastizing diagonal $D$---an object that in general admits no efficiently realizable representation (\Cref{def:eff_realizable}), whose recovery within natural circuit classes ranges from polynomial time (\Cref{thm:dstoq_efficient}) to $\NP$-complete (\Cref{thm:diag_stoq_hardness}), and whose existence is $\PSPACE$-complete to decide absent the VGP promise (\Cref{cor:pspace_stoquastizing_diagonal}).

\subsubsection{Entanglement and other phase-sensitive properties}
A stoquastizing  diagonal need not factorize across a bipartition, and a non-factorizing diagonal can change Schmidt values: a controlled-$Z$ across the cut is diagonal and maps the product state $\ket{+}\ket{+}$ to an entangled state.
Consequently, the entanglement of $\ket{\psi_0}$ across a cut can differ from that of the proxy ground-state
$\ket{\phi_0}$---the two states share magnitudes but not
entanglement structure---so entanglement spectra and related quantities are genuinely phase-sensitive.
Transition amplitudes $\braket{x|e^{-iHt}|y}$ likewise involve coherent summation over paths with phase factors determined by the full Hamiltonian, not just its stoquastic proxy, as does full state tomography.

\subsection{Summary: Spectral vs.\ State-level Equivalence}
\label{subsec:op_summary}

The operational boundary between VGP and stoquastic Hamiltonians, as far as we can currently delineate it, is as follows.
VGP and stoquastic Hamiltonians are equivalent---requiring no stoquastizing  diagonal---for ground-state energy, partition function and thermodynamic quantities, and diagonal equilibrium observables.
They appear to be separated---requiring phase data beyond the stoquastic proxy---for classical state description, classical state preparation, tomography, and phase-sensitive properties such as entanglement across a cut, with off-diagonal local observables an open intermediate case.
This is primarily because the stoquastizing diagonal of a VGP, local Hamiltonian can, in general, fail to admit an efficiently realizable representation.
VGP Hamiltonians are just as powerful as stoquastic Hamiltonians for spectral questions but appear genuinely distinct for state-level questions.

\section{VGP adiabatic quantum computer algorithms}
\label{sec:aqc}

A large body of work on adiabatic quantum computation (AQC) and quantum annealing is organized around the stoquastic/non-stoquastic divide.
Stoquastic anneals---transverse-field drivers with diagonal problem Hamiltonians---are widely regarded as classically tame: their instantaneous Hamiltonians are sign-problem-free, and their ground-state energy problem caps at $\StoqMA$~\cite{Bravyi_2006, Bravyi_2015, Bravyi_Hastings_2017}.
The universality constructions for AQC, by contrast, employ non-stoquastic Hamiltonians~\cite{Aharonov_2007, Biamonte_Love_2008}.
This contrast has motivated a sustained effort to engineer non-stoquastic drivers and catalysts---most prominently antiferromagnetic transverse couplers~\cite{Seki_Nishimori_2012}---on the premise that non-stoquasticity per se is the resource that makes a device hard to simulate, and therefore potentially more powerful (see Ref.~\cite{Albash_Lidar_2018} for a review).
The results of this paper sharpen that premise: to the extent that the argument for extra power runs through the sign problem or through ground-state energy complexity, the line in the sand is not stoquastic versus non-stoquastic but VGP versus non-VGP.

Our first observation is that the sign structure of an anneal is a property of its drivers alone, invariant under the choice of problem instance and annealing schedule.

\begin{proposition}[Schedule invariance of holonomy]
\label{prop:schedule_invariance}
Let
\begin{align*}
H(s)  =  H_{\mathrm{diag}}(s)  +  \sum_i A_i(s)\, O_i,
&& s\in[0,1],
\end{align*}
where $H_{\mathrm{diag}}(s)$ is diagonal, the $O_i$ are fixed off-diagonal operators with
\emph{pairwise-disjoint off-diagonal support}, and $A_i(s)>0$ for all $s$. Then the
transition graph $G_{H(s)}$ and the holonomy of every closed walk in it are independent
of $s$. In particular, $H(s)$ has VGP either for every $s$ or for no $s$, and which of
the two holds is decided by the off-diagonal skeleton $\sum_i O_i$ alone.
\end{proposition}

\begin{proof}
By pairwise disjointness of off-diagonal support, every nonzero off-diagonal entry of
$H(s)$ equals $A_i(s)$ times a fixed entry of a single $O_i$. Multiplication by the
positive scalar $A_i(s)$ changes neither which entries are nonzero nor their complex
phases, so $G_{H(s)}$ and all edge phases---hence all holonomies (\Cref{def:vgp})---are
independent of $s$. Diagonal terms contribute no edges.
\end{proof}

The standard annealing setting---a diagonal problem Hamiltonian and a fixed driver, possibly with catalysts, interpolated by positive schedules---satisfies these hypotheses, so the VGP question for an entire annealing path reduces to the VGP question for its drivers.
(If schedules are allowed to change sign, holonomies can depend on $s$ and VGP must be
assessed pointwise; the disjoint-support hypothesis is what rules out the other failure
mode, in which distinct drivers contribute to a common matrix element.)

Second, a VGP anneal is operationally a stoquastic anneal.
Suppose the driver skeleton admits a stoquastizing  diagonal $D$ (\Cref{prop:vgp_equiv_def}), and consider the conjugated path $H'(s) = DH(s)D^\dagger$, which by \Cref{prop:schedule_invariance} is stoquastic for every $s$.
The two paths have identical spectra and spectral gaps at every $s$, their instantaneous eigenstates are related by $D$, and---because $D$ is diagonal---their computational-basis statistics agree at all times: if $\ket{\psi'(0)} = D\ket{\psi(0)}$, then $\ket{\psi'(t)} = D\ket{\psi(t)}$ and $|\braket{z|\psi'(t)}|^2 = |\braket{z|\psi(t)}|^2$.
The initial condition is matched automatically, since the ground-state of the conjugated driver is $D$ applied to the ground-state of the original driver.
An annealer whose drivers are VGP but not stoquastic is therefore indistinguishable, at the level of computational-basis sampling, from a stoquastic annealer with the same schedules---whether or not anyone can exhibit $D$.
When no efficiently realizable $D$ exists (\Cref{ex:n_qubit_family_non_local_stoquastizing _diagonal}, \Cref{thm:vgp_hard_to_stoquastize}), the stoquastic twin admits no compact \emph{local} description---though its matrix entries remain efficiently computable, which is precisely what \Cref{thm:vgp_stoqma} exploits---and its ground-state energy problem remains $\StoqMA$-verifiable through the stoquastic proxy.

Third, VGP paths inherit the stoquastic complexity ceiling.
If every instantaneous Hamiltonian of an anneal has VGP, then in particular the ground-state energy problem of its final Hamiltonian lies in $\StoqMA$ (\Cref{thm:vgp_stoqma}), exactly as in the stoquastic case (\Cref{cor:stoqma_complete}).
The universality constructions for AQC~\cite{Aharonov_2007, Biamonte_Love_2008} rest on final Hamiltonians whose ground-state energy decision problem is hard for quantum polynomial-time computation; VGP paths cannot reproduce them unless such problems admit $\StoqMA$ verification.
Non-VGP instantaneous Hamiltonians are thus \emph{necessary} for the standard route to universal AQC, while non-stoquastic ones are not \emph{sufficient}.

These distinctions cut through well-studied driver designs.
The canonical non-stoquastic catalyst is the antiferromagnetic transverse coupler $+J X_iX_j$ with $J > 0$~\cite{Seki_Nishimori_2012}, and whether it takes an anneal across the VGP boundary depends entirely on context.
\begin{enumerate}
    \item[$(i)$] An exchange driver $J\sum_{\langle ij \rangle}(X_iX_j + Y_iY_j)$, $J>0$, on a \emph{bipartite} coupling graph is non-stoquastic, yet the sublattice gauge $D = \prod_{i \in A} Z_i$ stoquastizes it; this is precisely the Marshall sign rule~\cite{Marshall_1955}, read in our language as a stoquastizing diagonal.
    Such drivers---including the magnetization-conserving ``XY-mixers'' used for constrained optimization, whenever their coupling graphs are bipartite---are VGP, and by the discussion above they confer no advantage attributable to the sign problem.
    \item[$(ii)$] The same couplers combined with a transverse field $-\Gamma \sum_i X_i$, $\Gamma>0$, are genuinely non-VGP: the triangle $\ket{z} \to X_i\ket{z} \to X_iX_j\ket{z} \to \ket{z}$ in $G_H$ has entries $(-\Gamma)$, $(-\Gamma)$, $(+J)$ and holonomy $\Phi = (-1)^3 \cdot (+1) = -1$ (\Cref{def:vgp}).
    \item[$(iii)$] On non-bipartite coupler graphs, the couplers alone are already non-VGP: traversing an odd coupler cycle of length $\ell_c$ yields an induced cycle with all-positive entries and holonomy $(-1)^{\ell_c} = -1$; the same computation shows that complete-graph XY-mixers, which contain such triangles within each magnetization sector, are non-VGP even though even-length XY rings are VGP.
\end{enumerate}
It is the non-trivial holonomy cycle certificates of $(ii)$ and $(iii)$---not the non-positivity of any particular matrix entry---that license the claim that a device has crossed the sign problem boundary.
The empirical observation that the benefits attributed to non-stoquastic catalysts often survive ``de-signing,'' i.e., replacement of the catalyst by a sign-problem-free counterpart of similar structure~\cite{Crosson_2020}, is consistent with this picture: non-stoquasticity per se was never the operative resource.

For the design and benchmarking of annealers this suggests a discipline.
Before attributing power to a non-stoquastic driver, screen it for VGP: \Cref{prop:disjoint_eff_vgp_recognition} certifies VGP efficiently for drivers with disjoint PMR supports, and \Cref{cor:dstoq_ell} decides in polynomial time whether an efficiently realizable stoquastizing diagonal exists---recovering it in the YES case, at which point the device can be simulated sign-problem-free by PMR-QMC (\Cref{thm:qmc_boundary}) or replaced outright by its explicitly stoquastic twin.
Conversely, the absence of an evident stoquastizing diagonal is not evidence of non-VGP: recognition is $\coNP$-hard already for $2$-local Hamiltonians and $\PSPACE$-complete in general (\Cref{thm:2loc_vgp_hardness}, \Cref{thm:loc_vgp_hardness}), and there are VGP drivers with no efficiently realizable stoquastizing diagonal at all (\Cref{thm:vgp_hard_to_stoquastize}).
A claim of non-VGP, instead, should be provided by a proof of the non-existence of a stoquastizing diagonal or by an explicit non-trivial holonomy cycle, as in $(ii)$--$(iii)$ above.

Two caveats keep the line in the sand honest.
Crossing it is not sufficient for advantage: a non-VGP anneal defeats the closed-walk expansion of PMR-QMC, not every classical simulation method.
And remaining on the VGP side is not sufficient for simulability: sign-problem-freeness removes one obstruction, but sign-free QMC can still fail to equilibrate~\cite{Hastings_Freedman_2013}, a gap we record as the final open problem of \Cref{sec:discussion}.
What the VGP framework does settle is this: any argument for the power of a non-stoquastic annealer that proceeds from the sign problem, from ground-state energy complexity, or from the non-positivity of matrix elements applies only on the non-VGP side of the boundary---and whether a proposed device is on that side is a property of the cycles of its transition graph, not of its entries.

\section{Discussion}
\label{sec:discussion}

We summarize the main conclusions and state open problems.

\subsubsection{VGP as the PMR-QMC sign-problem-free boundary}
Within the PMR-QMC framework, stoquasticity is a basis-dependent sufficient condition for sign-problem-free simulation, while VGP is the exact, diagonal-similarity invariant necessary and sufficient condition (\Cref{thm:qmc_boundary}).
Two Hamiltonians related by a diagonal unitary share all holonomies and all QMC weights, yet one may be stoquastic and the other not.
VGP captures the physically invariant content.
Physical simulation questions should be indexed by invariants, not by merely sufficient---but not necessary---conditions such as stoquasticity.

\subsubsection{A computational wall between VGP and stoquasticity}
There exist VGP local Hamiltonians that belong to a family that admit efficient recognition of the VGP property, but are formally hard to stoquastize in a complexity-theoretic sense (\Cref{thm:vgp_hard_to_stoquastize}).
Thus VGP and stoquasticity are two distinct classes, computationally, but VGP more-accurately describes the intrinsic boundaries in Hamiltonian complexity that stoquasticity fails to completely capture. 

\subsubsection{Diagonal unitary transformations do not enlarge $\StoqMA$}
The local Hamiltonian problem restricted to VGP Hamiltonians is $\StoqMA$-complete (\Cref{thm:vgp_stoqma}, \Cref{cor:stoqma_complete}): the ground-state energy decision problem does not become harder when stoquasticity is hidden behind a diagonal unitary, however wild that diagonal may be.
The mechanism is worth isolating.
The stoquastic proxy map $H \mapsto \mathcal{S}(H)$ produces a cospectral stoquastic representative from any VGP instance using only local data; the reduction never verifies VGP, never constructs a stoquastizing unitary, and never constructs a compact operator-basis decomposition of $\mathcal{S}(H)$.
It uses the VGP promise solely to guarantee spectral equivalence between $H$ and $\mathcal{S}(H)$, and it needs only efficient query access to the matrix entries of $\mathcal{S}(H)$---not a local description of it (see~\Cref{lem:oracle-measurement}).
This yields the collapse $\VGPMA = \StoqMA$ (\Cref{cor:vgpma_collapse}).

\subsubsection{The equivalence extends to thermodynamics and diagonal observables}
Beyond energy, the partition function, all equilibrium thermodynamic quantities, and all diagonal equilibrium expectation values are \emph{identical} for a VGP Hamiltonian $H$ and its stoquastic proxy $\mathcal{S}(H)$ (\Cref{sec:operational}); whether general off-diagonal local observables likewise agree---and can be recovered from proxy data without the stoquastizing diagonal---is left open (\Cref{prob:offdiag_observables}).
The VGP--stoquastic distinction becomes operationally relevant for tasks involving the full quantum state: classical state description, state preparation, and phase-sensitive properties such as entanglement across a cut (\Cref{sec:operational}).
The recovery of an explicit stoquastizing diagonal, which is unnecessary for the spectral questions, is precisely what these state-level tasks require---and \Cref{sec:diagonal_recovery} shows that it is tractable under some natural conditions (\Cref{thm:dstoq_efficient}, \Cref{cor:dstoq_ell}) even though it is $\PSPACE$-complete in general (\Cref{cor:pspace_stoquastizing_diagonal}).

\subsubsection{Recognition is harder than exploitation}
A striking asymmetry runs through our results: one can \emph{use} the VGP promise far more easily than one can verify it (in the worst case).
Ground-state energy estimation under the promise is no harder than in the stoquastic case, yet deciding whether a given local Hamiltonian has VGP is $\coNP$-hard already at $2$-locality and $\PSPACE$-hard for $k \geq 5$ (\Cref{thm:2loc_vgp_hardness}, \Cref{thm:loc_vgp_hardness})---harder than the local Hamiltonian problem itself.
Nor does the promise deliver a usable stoquastizing diagonal: there are VGP local Hamiltonians admitting no efficiently realizable stoquastizing  diagonal at all (\Cref{thm:vgp_hard_to_stoquastize}), and conversely there are families where VGP is efficiently certifiable while stoquasticity recognition is $\coNP$-hard (\Cref{cor:recognition_vgp_easy_stoq_hard}).
Stoquasticity and VGP are thus not merely different classes but differently \emph{shaped} ones, and the practical consequence is that one should screen for VGP with the efficient sufficient conditions of \Cref{prop:disjoint_eff_vgp_recognition} and \Cref{thm:dstoq_efficient} rather than search for a stoquastizing diagonal.
We leave the identification of other Hamiltonian families that admit efficient recognition of the VGP property for future work.

\subsubsection{Consequences for adiabatic quantum computation}
The program of engineering non-stoquastic drivers and catalysts rests on the premise that non-stoquasticity is itself the resource that makes a device hard to simulate.
Our results relocate that premise (\Cref{sec:aqc}): for standard anneals the holonomies are schedule- and instance-invariant (\Cref{prop:schedule_invariance}), so the relevant property is a feature of the drivers alone; a VGP anneal is indistinguishable from a stoquastic anneal at the level of computational-basis sampling, whether or not the stoquastizing diagonal is efficiently realizable; and VGP paths inherit the $\StoqMA$ ceiling.
The line in the sand is VGP versus non-VGP, not stoquastic versus non-stoquastic.
Concretely, bipartite exchange drivers and bipartite XY-mixers are VGP by the Marshall sign rule~\cite{Marshall_1955} and confer no sign-problem-derived advantage, whereas antiferromagnetic couplers combined with a transverse field, or placed on non-bipartite graphs, carry explicit non-trivial holonomy cycles and do cross the boundary.

\subsubsection{Scope and limitations}
Three limitations bound these conclusions.
First, our notion of sign-problem-freeness in classical simulation is specific to the PMR-QMC framework; we have not shown that VGP characterizes the sign problems arising in worldline, stochastic series expansion, or determinantal formulations, which may have different tractability boundaries.
Second, our completeness results concern the local Hamiltonian problem \emph{restricted to the VGP class}: widening the promise from stoquastic to VGP leaves the ground-state energy problem $\StoqMA$-complete (\Cref{cor:stoqma_complete}). 
This should not be read as an instance-wise dividing line between $\StoqMA$ and $\QMA$.
We do not show that non-VGP Hamiltonians fall outside $\StoqMA$---many non-VGP families are trivially easy---so VGP is not a boundary separating $\StoqMA$-easy from $\QMA$-hard instances.
What the results establish is that $\StoqMA$-completeness is retained across the entire VGP class: stoquasticity was never the operative feature distinguishing $\StoqMA$ from $\QMA$, and the diagonal-unitary-invariant class VGP inherits the same complexity.
In this sense the traditional attribution of the $\StoqMA$ vs. $\QMA$ boundary to stoquasticity is too narrow; it is a statement about vanishing versus non-vanishing geometric phase.
Third, sign-problem-freeness removes one obstruction to classical simulation, not all of them: a VGP Hamiltonian might still be hard to sample from in PMR-QMC (\Cref{prob:qmc_efficiency}).

\subsubsection{Conclusion}
The results of this paper can be summarized as follows.
The PMR-QMC sign problem boundary, established in~\cite{Hen_2021}, is that $H$ is sign-problem-free under PMR-QMC if and only if $H$ has VGP.
The complexity statement, established here, is that the local Hamiltonian problem is $\StoqMA$-complete under a VGP promise---and hence that hiding stoquasticity behind a diagonal unitary, even a non-representable one, confers no additional verification power.
The operational boundary, delineated here, is that VGP and stoquastic Hamiltonians are equivalent for spectral, thermodynamic, and diagonal-observable questions but appear distinct for state-level questions, with the stoquastizing  diagonal marking the divide.
Along the way we classified the complexity of recognizing VGP and of recovering stoquastizing  diagonals, exhibited VGP Hamiltonians that are provably hard to stoquastize, and drew out the consequences for the design of adiabatic quantum algorithms.
We hope that these results can be used to re-examine any potential separation between $\StoqMA$ and $\QMA$ through the lens of invariance under diagonal unitary transformation: the geometric phases of cycles in the transition graph, not the sign of any particular matrix element, is what one should be looking through.

\subsubsection{Open problems.}

\begin{problem}[$\coNP$-completeness of VGP recognition for 2-local Hamiltonians]
    \label{prob:2_local_poly_diam}
    We have shown that VGP recognition is $\coNP$-hard for 2-local spin-1/2 Hamiltonians.
    Is this decision problem also in $\coNP$?
    A natural approach to proving $\coNP$ membership is bounding the sizes of isometric cycles in the transition graph of an arbitrary 2-local spin-1/2 Hamiltonian $H$.
    If one can prove that the isometric cycles in $G_H$ are bounded in size by a polynomial, then $\coNP$ membership follows immediately.
    If one can provide an example where there are isometric cycles of superpolynomial length in $G_H$, then we conjecture that VGP recognition in 2-local Hamiltonians is harder than $\coNP$-hard.
\end{problem}

\begin{problem}[SAT encoding of 2-local VGP recognition]
    \label{prob:2local_vgp_recog_sat_encoding}
    If VGP recognition is in $\coNP$ for 2-local Hamiltonians, then one can construct, in polynomial time, a SAT encoding of the complementary recognition problem (i.e., the formula is satisfiable exactly when the input does not have VGP).
    Thus VGP recognition can be decided by checking unsatisfiability of this encoding.
    If one can find a useful encoding, then modern SAT solvers~\cite{Biere_2021} may often decide such instances efficiently in practice.
\end{problem}

\begin{problem}[Accessibility of off-diagonal observables]
    \label{prob:offdiag_observables}
    Under a VGP promise, all diagonal equilibrium observables are computable from the stoquastic proxy (\Cref{sec:operational}).
    For an off-diagonal $k$-local observable $O$, the expectation value requires the relative phases $\theta_x - \theta_y$ on the pairs $(x,y)$ where $O$ is supported.
    When those pairs are edges of $G_H$, the phases are locally computable from the Hamiltonian; in general they are path products in $G_H$, and the connecting paths need not be short.
    For which classes of $(H, O)$ can $\braket{\psi_0|O|\psi_0}$ be estimated from stoquastic-proxy data without recovering a stoquastizing  diagonal?
\end{problem}

\begin{problem}[Non-constructive sign-curing maps]
    \label{prob:sign_curing_maps}
    The stoquastic proxy map $\mathcal{S}$ (\Cref{def:stoq_proxy_map}) is not efficiently realizable in general.
    That is, the Pauli expansion of $\mathcal{S}(H)$ may be non-local and have exponentially many non-zero terms, even when $H$ is local (\Cref{ex:n_qubit_family_non_local_stoquastizing _diagonal}).
    However, $\mathcal{S}(H)$ can still have its ground-state energy estimated by a $\StoqMA$ verifier~(see discussion of \Cref{lem:stoq_map_decomp}).
    Are there other maps on Hamiltonians that fail to preserve locality, but allow one to obtain efficient access to the matrix entries of a cospectral stoquastic Hamiltonian $H_{\stoq}$?
    And, in this case, can one prove that the local Hamiltonian problem for $H_{\stoq}$ is in $\StoqMA$ using~\Cref{lem:oracle-measurement}?
\end{problem}

\begin{problem}[QMC efficiency beyond the sign problem]
    \label{prob:qmc_efficiency}
    VGP guarantees sign-problem-free simulation, but sign-freeness is necessary, not sufficient, for efficient Monte Carlo.
    Are there VGP families where the effective Monte Carlo cost is qualitatively worse than for explicitly stoquastic Hamiltonians with similar physical parameters?
\end{problem}

\section*{Acknowledgments}

We thank Milad Marvian and Michael Jarret for useful discussions. 

\bibliography{refs}
\appendix

\onecolumngrid

\section{Complexity Preliminaries}
\label{app:complexity_class_definitions}

\begin{definition}
    $\P$ is the class of decision problems (or languages) that can be solved by a deterministic Turing machine in polynomial time, i.e., $O(n^k)$ for some constant $k$.
    \qed
\end{definition}
\begin{definition}
    $\NP$ (Nondeterministic Polynomial time) is the class of decision problems for which a YES answer can be verified by a deterministic Turing machine in polynomial time, given a polynomial-sized certificate (or witness). 
    Equivalently, it is the class of languages decidable by a nondeterministic Turing machine in polynomial time.

    $\coNP$ is the class of decision problems whose complements are in $\NP$. Alternatively, it is the class of problems for which a NO answer can be verified by a deterministic Turing machine in polynomial time given a short certificate of rejection.
    \qed
\end{definition}
\begin{definition}
    $\PSPACE$ is the class of decision problems that can be solved by a deterministic (or nondeterministic) Turing machine using a polynomial amount of memory space, i.e., $O(n^k)$ space for some constant $k$, regardless of the time computation takes.
    \qed
\end{definition}
\begin{definition}
    $\SigmaTwoP$ is the second level of the polynomial hierarchy, formally defined as $\SigmaTwoP = \NP^{\NP}$. \qed
\end{definition}
An oracle Turing machine $M^A$ is a standard Turing machine equipped with a special query tape that can query membership in a language $A$ (the ``oracle'') in a single computation step. 
\begin{itemize}
    \item $\NP^{\NP}$ represents the class of problems decidable in polynomial time by a \emph{nondeterministic} Turing machine with access to an oracle for an $\NP$-complete problem (like SAT).
    \item $\NP^{\coNP}$ represents the class of problems decidable by a nondeterministic polynomial-time Turing machine with access to a $\coNP$ oracle.
\end{itemize}
Because an oracle machine can simply negate the answer from its oracle at no extra cost, an oracle for $\NP$ provides the exact same computational power as an oracle for $\coNP$. Therefore, $\SigmaTwoP = \NP^{\NP} = \NP^{\coNP}$.

We will now define the Hamiltonian complexity classes relevant to this paper.

\begin{definition}
    \label{def:stoqma}
    Let $a,b:\mathbb{N}\to[0,1]$ be efficiently computable functions satisfying $a(n)-b(n)\ge 1/\poly(n)$. A promise problem $L=(L_{\text{yes}},L_{\text{no}})$ is in $\StoqMA(a,b)$ if there exists a polynomial-time uniform family of quantum circuits $\{V_x\}$ with the following form.
    
    For an input $x\in\{0,1\}^n$, the verifier $V_x$ acts on a witness register $W$ of $\poly(n)$ qubits, a register of ancilla qubits initialized to $\ket{0}$, and a register of ancilla qubits initialized to $\ket{+}$.
    The circuit $V_x$ is a classical reversible circuit; equivalently, it is generated by gates from
    \begin{align*}
        \{X,\mathrm{CNOT},\mathrm{Toffoli}\}.
    \end{align*}
    At the end of the computation, a designated output qubit is measured in the Hadamard basis $\{\ket{+},\ket{-}\}$, and the verifier accepts upon outcome $\ket{+}$.
    
    The following completeness and soundness conditions hold:
    \begin{enumerate}[label=(\roman*)]
        \item If $x\in L_{\text{yes}}$, then there exists a quantum witness $\ket{\psi}$ such that
        \begin{align*}
            \Pr[V_x \text{ accepts } \ket{\psi}] \ge a(n).
        \end{align*}
        \item If $x\in L_{\text{no}}$, then for every quantum witness $\ket{\psi}$,
        \begin{align*}
            \Pr[V_x \text{ accepts } \ket{\psi}] \le b(n).
        \end{align*}
    \end{enumerate}
    We define $\StoqMA$ to be the union of $\StoqMA(a,b)$ over all such inverse-polynomially separated parameters $a$ and $b$. \qed
\end{definition}

\begin{definition}
    \label{def:ma}
    $\MA$ (Merlin--Arthur) is the class of promise problems for which a YES answer can be verified by a probabilistic polynomial-time classical verifier given a polynomial-sized classical witness.
    More precisely, there is a polynomial-time randomized verifier such that: 
    \begin{itemize} 
        \item For every YES instance, there exists a classical witness that causes the verifier to accept with probability at least $2/3$.
        \item For every NO instance, every classical witness causes the verifier to accept with probability at most $1/3$. 
    \end{itemize}
    The constants $2/3$ and $1/3$ may be replaced by any efficiently computable completeness and soundness parameters separated by at least an inverse polynomial. \qed 
\end{definition}

\begin{definition} 
    \label{def:qma}
    $\QMA$ (Quantum Merlin--Arthur) is the quantum analogue of $\MA$.
    It is the class of promise problems for which a YES answer can be verified by a polynomial-time quantum verifier given a polynomial-sized quantum witness.
    More precisely, there is a polynomial-time uniform family of quantum circuits such that: \begin{itemize} 
        \item For every YES instance, there exists a quantum witness that causes the verifier to accept with probability at least $2/3$.
        \item For every NO instance, every quantum witness causes the verifier to accept with probability at most $1/3$.
    \end{itemize} 
    As with $\MA$, the completeness and soundness parameters may be replaced by any efficiently computable parameters separated by at least an inverse polynomial. \qed 
\end{definition}

The complexity classes sit in the following hierarchies:
\begin{center} $\P \subseteq \NP \subseteq \MA \subseteq \StoqMA = \VGPMA \subseteq \QMA \subseteq \PSPACE$,\\ $\P \subseteq \NP \subseteq \SigmaTwoP \subseteq \PSPACE$. \end{center}

\section{Proof of~\Cref{lem:oracle-measurement}}
\label{app:proof_of_oracle_measurement}

The proof has two steps.
The first records that the measurement primitive of Ref.~\cite{Ioannou_2020} applies to any diagonal projector supplied as a reversible circuit, and that locality plays no role in it.
The second reduces a general nonnegative diagonal to that case by binary expansion, and charges the resulting truncation error against the promise gap.

\subsection{Diagonal projectors given by reversible circuits}

Let $\Pi = \sum_{y \in S}\ket{y}\!\bra{y}$ be a diagonal projector whose indicator $y \mapsto [\,y \in S\,]$ is computed by a reversible circuit $C_\Pi$ over $\{X,\mathrm{CNOT},\mathrm{Toffoli}\}$ of size $\poly(n)$.
The measurement primitive of Ref.~\cite{Ioannou_2020} measures the term $X \otimes \Pi$ as follows: apply $C_\Pi$ to write the indicator onto a flag qubit, introduce one $\ket{+}$ ancilla, read the designated output qubit in the Hadamard basis, and uncompute the work register of $C_\Pi$.
The resulting acceptance probability is affine in $\braket{\psi|X\otimes\Pi|\psi}$, with coefficients independent of the witness.

The point to record is that this primitive consumes $\Pi$ \emph{only} through the circuit $C_\Pi$.
In Ref.~\cite{Ioannou_2020} the projectors arise from local classical Hamiltonians, and locality is invoked solely to guarantee that a circuit of size $\poly(n)$ exists; it is used nowhere else in the argument.
When such a circuit is supplied directly, as it is below, the primitive applies verbatim.

\subsection{General nonnegative diagonals by binary expansion}

Write $G := F/M$, so that $0 \le G(y) \le 1$ for every $y$, and let $\widehat G(y) := \widehat F(y)/M$ be the truncation produced by $\mathcal{O}_F$, which satisfies $|\widehat G(y) - G(y)| \le 2^{-b}$.
Expanding $\widehat G$ in its bits gives
\begin{align*}
    \widehat G = \sum_{\ell=1}^{b} 2^{-\ell}\,\Pi_\ell,
    && \Pi_\ell := \!\!\sum_{y \,:\, \ell\text{-th bit of }\widehat G(y) \,=\, 1}\!\! \ket{y}\!\bra{y},
\end{align*}
a nonnegative combination of $b$ diagonal projectors.
Each indicator $y \mapsto [\,\ell\text{-th bit of }\widehat G(y)\,]$ is a Boolean function of the output of $\mathcal{O}_F$, and is therefore computed exactly by a reversible circuit of size $\poly(n,b)$: run $\mathcal{O}_F$, copy the relevant output bit onto a flag qubit, and uncompute the work register.
The previous step therefore applies to every $\Pi_\ell$.

Sample $\ell \in \{1,\ldots,b\}$ with probability proportional to $2^{-\ell}$, using fresh ancillas exactly as the sum over decomposition terms is handled in Ref.~\cite{Ioannou_2020}, and run the measurement above for the sampled $\Pi_\ell$.
The acceptance probability of the resulting verifier is affine in
\begin{align*}
    \sum_{\ell} 2^{-\ell}\,\braket{\psi|X\otimes\Pi_\ell|\psi}  \propto  \braket{\psi|X\otimes\widehat G|\psi}.
\end{align*}

It remains to charge the truncation error.
Since $\widehat G$ and $G$ are simultaneously diagonal,
\begin{align*}
    \big|\braket{\psi|X\otimes\widehat G|\psi}-\braket{\psi|X\otimes G|\psi}\big|
     \le \big\|X\otimes(\widehat G-G)\big\|
     = \max_y\big|\widehat G(y)-G(y)\big| \le 2^{-b}.
\end{align*}
Taking $b=\lceil\log_2(M/\delta)\rceil+O(1)$ therefore makes the induced error in $p(\psi)$ at most $\delta$, as claimed.
Every gate used lies in $\{X,\mathrm{CNOT},\mathrm{Toffoli}\}$ and every ancilla is prepared in $\ket{0}$ or $\ket{+}$, so the circuit is a valid $\StoqMA$ verifier.
The case of $F$ in place of $X \otimes F$ is identical, dropping the flag qubit.

\qed

\section{Proof of~\Cref{thm:2loc_vgp_hardness}}
\label{app:proof_of_twolocalvgp_conphard}
To prove this theorem, we reduce from an instance of the $\NP$-complete $\PARTITION$ decision problem~\cite{Lucas_2014}.

\begin{definition}
    The $\PARTITION$ decision problem is as follows:
    
    \begin{itemize}[leftmargin=5.5em]
        \item[\textbf{Input}:] Suppose $a_1, \ldots, a_n$ are positive integers.
        \item[\textbf{Problem}:] Decide if there is an assignment of signs $s_1,\ldots,s_n \in \{\pm 1\}^n$ such that $ s_1a_1 + \ldots + s_na_n = 0$. \qed
    \end{itemize}
\end{definition}

Now we prove~\Cref{thm:2loc_vgp_hardness}.
    We use Ref.~\cite{Lucas_2014}'s encoding of an instance of $\PARTITION$ into an Ising Hamiltonian. Let $a_1, \ldots, a_n$ be the given positive integers, introduce qubits $1, 2$ and $3, \ldots, n+2$, and set
    \begin{align*}
        H_0 := \sum_{i=1}^{n} a_i Z_{i+2}
    \end{align*}
    then put
    \begin{align*}
        H_{+} := H_0 + \frac{1}{2}I && \text{and} && H_{-} := H_0 - \frac{1}{2}I.
    \end{align*}
    Now consider the Hamiltonian
    \begin{align*}
        H = -( (H_{+}) X_{1} + (H_{-})X_2 + X_1X_2).
    \end{align*}
    By the induced-cycle characterization of VGP in \Cref{prop:vgp_equiv_def}(iii), it is sufficient to show that the weight of every induced cycle of $G_H$ has trivial complex phase if and only if there is no assignment of signs $\{s_i\}$ making $\sum_{i=1}^{n} s_ia_i = 0$.
    
    All PMR terms of $H$ act non-trivially on qubits $1$ and $2$ only, with $H_{\pm}$ diagonal on the remaining qubits.
    Consequently, every connected component of $G_H$ is a complete graph $K_4$ on the four assignments of qubits $1, 2$, with the spectator qubits fixed to some configuration $z$; within such a component, every edge of a given type ($X_1$, $X_2$, or $X_1X_2$) carries the same weight (\Cref{def:transition_graph}), namely $H_+(z)$, $H_-(z)$, and $1$, respectively, where $H_\pm(z) := \bra{z}H_0\ket{z} \pm \tfrac12$.
    (All edges are present: $H_\pm(z) \neq 0$ since $\bra{z}H_0\ket{z}$ is an integer.)
    The induced cycles of $K_4$ are its triangles, and every triangle uses each edge type (PMR term) exactly once, so the weight of any induced cycle in the component labeled by $z$ is
    \begin{align*}
        w(z) = H_+(z)\, H_-(z) = \bra{z}\, H_0^2 - \frac{1}{4}I \,\ket{z}.
    \end{align*}
    Equivalently, in the fundamental-cycle language of \Cref{def:fundamental_cycle}, the triangles arise from the product of the three PMR terms of $H$, and the sign-adjusted product
    \begin{align*}
        (-1)^3(-H_{+} X_1)(-H_{-}X_2)(-X_1X_2)
        = H_{+} H_{-} = H_0^2 - \frac{1}{4}I
    \end{align*}
    encodes exactly these weights on its diagonal.
    Now observe that $\bra{z} H_0^2 \ket{z}$ is either 
    \begin{itemize}[leftmargin=2em]
        \item a positive integer
        \item or 0, which happens if and only if
        $\sum_{i=1}^n s_ia_i = 0$ where
        \begin{align*}
            s_i = \begin{cases}
                +1 & \text{if bit } i+2 \text{ of } z \text{ is } 0,\\
                -1 & \text{if bit } i+2 \text{ of } z \text{ is } 1.
            \end{cases}
        \end{align*}
    \end{itemize}
    It follows that 
    \begin{align*}
        w(z) = \braket{z|\,H_0^2 - \frac{1}{4}I \,|z} < 0 \iff \bra{z} H_0^2 \ket{z} = 0,
    \end{align*}
    and otherwise $w(z) \geq 3/4$, since $\bra{z}H_0\ket{z}$ is an integer.
    That is, the induced-cycle weight in the component labeled by $z$ has non-trivial (negative) phase if and only if $\bra{z} H_0^2 \ket{z} = 0$. Hence $H$ \emph{does not} have VGP if and only if there exists a computational basis state $\ket{z}$ so that $\bra{z} H_0^2 \ket{z} = 0$, the existence of which would settle this instance of the $\PARTITION$ decision problem.
    The construction of $H$ is efficient, so deciding if such a computational basis state exists is $\NP$-hard.
    Thus the converse problem---deciding if $H$ has VGP---is $\coNP$-hard.

\qed

\section{Proof of~\Cref{thm:loc_vgp_hardness}}
\label{app:proof_of_threelocalvgp_pspacecomplete}
We first argue $\PSPACE$ membership.
A Hamiltonian fails to have VGP precisely when $G_H$ contains a closed walk of non-trivial holonomy, and a witnessing walk may be exponentially long, so it must be searched rather than merely evaluated.
By the assumption of~\Cref{sec:preliminaries}, every non-zero entry may be written as $H_{xy} = R_{xy}\,\omega^{k_{xy}}$ with $R_{xy} > 0$, $\omega = e^{2\pi i/m}$, and $k_{xy} \in \Z_m$ computable in $\poly(n)$ time from the local description of $H$.
Taking $m$ even without loss of generality, so that $-1 = \omega^{m/2}$, the holonomy of a closed walk of length $L$ is
\begin{align*}
    \Phi(\gamma) = \omega^{\,(m/2)L \,+ \, \sum_t k_t},
\end{align*}
so the entire holonomy is tracked by a single accumulated exponent modulo $m$---an integer of $\log_2 m = \poly(n)$ bits, however long the walk.
Moreover, a witnessing walk may be taken to be a simple cycle, hence of length at most $2^n$.
Non-VGP detection is therefore in $\mathsf{NPSPACE}$: nondeterministically guess a start vertex and walk $G_H$ edge by edge for at most $2^n$ steps, storing only the current vertex, a step counter, and the accumulated exponent modulo $m$, and accept upon returning to the start with accumulated exponent $\not\equiv 0$.
By Savitch's theorem $\mathsf{NPSPACE} = \PSPACE$~\cite{Savitch_1970}, and $\PSPACE$ is closed under complement, so VGP recognition is in $\PSPACE$.
It suffices to prove the problem is $\PSPACE$-hard.

The rest of the proof of~\Cref{thm:loc_vgp_hardness} will be approached as follows:
\begin{enumerate}[label=\textbf{\arabic*.}]
    \item We will introduce \emph{nondeterministic constraint logic} (NCL), a combinatorial system used as a formal computational model.
    An NCL instance models computation via moves in its \emph{reconfiguration graph}, which contains as vertices \emph{legal configurations} of the NCL instance.
    \item Given an arbitrary NCL instance $\Gamma$, we will construct a 5-local, geometrically local Hamiltonian $H$ such that some collection of connected components of $G_H$ is isomorphic to the reconfiguration graph of $\Gamma$.
    We also construct $H$ to be stoquastic (entry-wise non-positive), so it has VGP.
    \item We introduce a known $\PSPACE$-complete problem pertaining to NCL called $\FREENCLREV$ (see \Cref{lem:hard_decision_problem_in_ncl}).
    Given an arbitrary instance of $\FREENCLREV$ and its verifier $V$, we will introduce a 5-local gadget Hamiltonian $H_{\text{gadget}} \geq 0$ (entry-wise) so that $H + H_{\text{gadget}}$ does not have VGP if and only if $V(\Gamma)$ accepts.
    This proves VGP recognition is $\PSPACE$-hard.
\end{enumerate}

\subsection{Nondeterministic Constraint Logic (NCL)}
Nondeterministic Constraint Logic (NCL), introduced by Hearn and Demaine (Ref. \cite{Hearn_2004}), is a graph-based model of reconfiguration. 
An NCL instance consists of an edge-weighted graph together with an orientation of its edges. 
A configuration is legal if, at every vertex, the total weight of incoming edges is at least a prescribed minimum inflow, usually 2. 
A reconfiguration move consists of reversing the orientation of a single edge, subject to the constraint that both the starting and resulting orientations are legal. 
The reachability question asks whether one legal orientation can be transformed into another by such a sequence of moves.

In the standard AND/OR formulation, the graph has maximum degree 3 and uses only red and blue edges, with red edges having weight 1 and blue edges having weight 2. 
Every vertex has minimum inflow requirement of 2. An OR vertex is incident to three blue edges, so any one incoming edge suffices to satisfy the inflow constraint. 
An AND vertex is incident to two red edges and one blue edge; the blue edge may be directed outward only when both red edges are directed inward. 
Despite this very restricted local structure, deciding if you can (legally) reverse edges to go from a start configuration to a target configuration is $\PSPACE$-complete, making AND/OR NCL a useful source problem for reductions to reconfiguration problems. 
$\PSPACE$-completeness also survives if the NCL constraint graph has \emph{bounded bandwidth}~\cite{vanderZanden_2015}.
\begin{definition}
    A graph has \emph{bounded bandwidth} if its vertices (and thus its edges) can be linearly ordered such that the maximum index difference between any connected vertices is strictly limited by a constant. \qed
\end{definition}
This is analogous to geometric locality for Hamiltonians.
As we will see, one can indeed embed an NCL reconfiguration graph into a geometrically local Hamiltonian's transition graph.

\subsection{Encoding the reconfiguration graph of an NCL instance into a 5-local Hamiltonian}

Suppose $\Gamma = (V,E)$ is an $\ANDOR$ NCL instance with bounded bandwidth. 
For every edge $e$, introduce a new qubit $b(e)$.
For every vertex $w$ with incident edge $e$, define $\sigma(e,w)$ to be 
\begin{itemize}
    \item $+1$ if $b(e) = 0$ implies $e$ is pointing toward $w$.
    \item $-1$ otherwise.
\end{itemize}
Now, suppose $e=\{u,v\} \in E$ is arbitrary.
Every vertex in $\Gamma$ has degree 3, so for each $w \in \{u,v\}$, let $e_1^{(w)} $ and $e_2^{(w)}$ denote the edges---different from $e$---that are incident to $w$.
Define the following indicator:
\begin{align*}
    \ind_{w, i} = \frac{I + \sigma(w,e_i^{(w)})\, Z_{b(e_i^{(w)})}}{2}.
\end{align*}
In words: $\bra{z}\ind_{w, i}\ket{z}$ is 1 when $e_i^{(w)}$ is oriented \emph{toward} $w$ (i.e., providing inflow to $w$), and $0$ otherwise.

We construct a Hamiltonian term $h_e$ by performing a case analysis on $e$:
\begin{enumerate}
    \item Assume that $e$ is a blue edge.
    There are three cases:
    \begin{enumerate}
        \item[a.] Assume $u$ and $v$ are both $\OR$ vertices.
        Then every edge incident to $u$ or $v$ is blue.
        Define
        \begin{align*}
            h_e  := \left[\ind_{u,1} + \ind_{u, 2}  \right]\left[\ind_{v, 1} + \ind_{v, 2}  \right]
        \end{align*}
        \item[b.] Assume that $u$ and $v$ are both $\AND$ vertices.
        Define
        \begin{align*}
            h_e = \ind_{u,1}\ind_{u,2}\ind_{v,1}\ind_{v,2}
        \end{align*}
        \item[c.] Without loss of generality, assume that $u$ is an $\AND$ vertex and $v$ is an $\OR$ vertex.
        Define
        \begin{align*}
            h_e &:= \ind_{u,1}\ind_{u,2}[\ind_{v,1} + \ind_{v,2} - \ind_{v,1}\ind_{v,2}]
        \end{align*}
    \end{enumerate}
    \item Assume that $e$ is a red edge.
    Both $u$ and $v$ are $\AND$ vertices.
    Without loss of generality, assume that $e_1^{(u)}$ is blue and $e_1^{(v)}$ is blue.
    Define
    \begin{align*}
        h_e := \ind_{u, 1}\ind_{v, 1}.
    \end{align*}
\end{enumerate}
One can verify that the Hamiltonians are 5-local and are entry-wise $\geq 0$.
Defining $h_e$ for every $e \in E$, the stoquastic 5-local Hamiltonian that encodes the reconfiguration graph of $\Gamma$ is
\begin{align*}
    H = \sum_{e} -h_eX_{b(e)}.
\end{align*}
Notice that $H$ is geometrically local since the NCL instance has bounded bandwidth.
Recall that the condition for the legality of the reversal of edge $e = \{u,v\}$ is that $u$ and $v$ must have their minimum inflow constraint satisfied \emph{no matter what orientation} $e$ is in.
We analyze how each term enforces this legality condition for edge reversal, case by case:
\begin{enumerate}
    \item[1.a.] Every relevant edge in this case is blue, i.e. at least one edge oriented toward $u$ and at least one edge oriented toward $v$ is sufficient and necessary to satisfy the inflow requirements of both $u$ and $v$ without $e$.
    That is exactly what $h_e$ checks in this case.
    \item[1.b.] Reversing $e$ is legal if and only if every red edge incident to $u$ and $v$ is providing inflow to $u$ and $v$.
    \item[1.c.]
    Reversing $e$ is legal if and only if $u$ is provided inflow from both of its red edges and $v$ is provided inflow by at least one of its blue edges.
    \item[2.] Reversing $e$ is legal if and only if the blue edges incident to $u$ and $v$ are providing inflow to $u$ and $v$.
    That is the only way for $u$ and $v$ to have their minimum inflow requirement met without $e$.
\end{enumerate}
Therefore, given a legal configuration $z \in \{0,1\}^{|E|}$, a transition $\ket{z} \to X_{b(e)}\ket{z}$ is allowed in $G_H$ if and only if the edge $e$ can legally be reversed in configuration $z$ of $\Gamma$.
Therefore, connected components of $G_H$ either contain only legal configurations of $\Gamma$ or none at all.

\subsection{$\PSPACE$-hardness reduction}

\begin{lemma}[$\FREENCLREV$ from Ref.~\cite{DeBiasi_2016}]           \label{lem:hard_decision_problem_in_ncl}
    Suppose $\Gamma = (V,E)$ is an $\ANDOR$ NCL instance with bounded bandwidth.
    The following decision problem is $\PSPACE$-complete:

    Suppose $e_A$ and $e_B$ are red edges in $E$.
    Decide if there exist mutually reachable configurations $x$ and $y$ such that reversing $e_A$ is legal from configuration $x$ and reversing $e_B$ is legal from configuration $y$.
\end{lemma}
We prove~\Cref{thm:loc_vgp_hardness} by reducing from $\FREENCLREV$.
Because our 5-local construction allows edge reversals in illegal configurations, we will handle illegal states by guaranteeing that the gadget Hamiltonian we construct \emph{does not allow} illegal configurations to be part of non-VGP cycles.
We prove this using a key fact about the behavior of the illegal-configuration connected components, which we formulate in the following lemma:
\begin{lemma}[Locked edge lemma]
    \label{lem:edge_locking_lemma}
    Suppose $x$ is an illegal configuration for $\Gamma$.
    That is, some vertex $u$ does not have its minimum inflow constraint satisfied under the configuration $x$ of $\Gamma$.
    Then every edge $e_1,e_2,e_3$ incident to $u$ cannot be reversed under the transition rules imposed by $H$.
    Equivalently, the transitions $X_{b(e_1)}$, $X_{b(e_2)}$, and $X_{b(e_3)}$ are not allowed in the connected component of $G_H$ containing $\ket{x}$.
    Consequently, $u$ does not have its minimum inflow constraint satisfied under any configuration $z$ in the connected component of $G_H$ containing $\ket{x}$.
\end{lemma}
\begin{proof}
    In each case of the definition of $h_e$, the indicator factors attached to an edge $e$ incident to $u$ demand that $u$'s \emph{other} incident edges alone provide inflow at least $2$: for an $\OR$ vertex, at least one other blue edge must point toward $u$, while for an $\AND$ vertex, both red edges (when $e$ is the blue edge) or the blue edge (when $e$ is red) must point toward $u$.
    In every case this inflow alone satisfies $u$'s minimum inflow constraint, contradicting the deficiency of $u$ under $x$.
    Hence the factors related to $u$ in $h_e$ vanish on $\ket{x}$ for every $e$ incident to $u$, the transitions $X_{b(e_i)}$ are disallowed, and, since the orientations of $u$'s edges are thereby frozen, the deficiency of $u$ persists throughout the connected component containing $\ket{x}$.
\end{proof}

Now we define the gadget Hamiltonian.
First, partition $V = V_{\OR} \cup V_{\AND}$ into two sets containing $\OR$ vertices and $\AND$ vertices, respectively.
For every vertex $v \in V_\OR$, let $e_1^{(v)},e_2^{(v)},$ and $e_3^{(v)}$ denote its 3 blue edges and for every $v \in V_\AND$ let $e_1^{(v)},e_{2}^{(v)},$ and $e_{3}^{(v)}$ denote its one blue edge and two red edges, respectively.
Create an ordering on the vertices in the partition:
\begin{align*}
    V_\OR = \{\alpha_0, \ldots, \alpha_\ell\} && V_{\AND} = \{ \beta_0, \ldots, \beta_m \}.
\end{align*}
Now introduce ancilla qubits
\begin{align*}
    p_0, \ldots, p_\ell && \text{and} && q_0, \ldots, q_{2m+1}
\end{align*}
and two more qubits $a$ and $b$.
Each ancilla can be identified with some vertex in the constraint graph.

We define the following 5-local Hamiltonian on the new qubits:
\begin{align*}
    H_{\gadget} &= \sum_{i=0}^{\ell-1} [\ind_{\alpha_i,1} + \ind_{\alpha_i,2} + \ind_{\alpha_i,3}]X_{p_{i}}X_{p_{i+1}}\\
    &+ [\ind_{\alpha_\ell,1} + \ind_{\alpha_\ell,2} + \ind_{\alpha_\ell,3}] X_{p_\ell}X_{q_0}\\
    &+ \sum_{i=0}^{m-1} \left[ \ind_{\beta_i,1} + \ind_{\beta_i,2} \right]X_{q_{2i}}X_{q_{2i+1}} + \left[ \ind_{\beta_i,1} + \ind_{\beta_i,3} \right]X_{q_{2i+1}}X_{q_{2i+2}}\\
    &+  \left[ \ind_{\beta_m,1} + \ind_{\beta_m,2} \right]X_{q_{2m}}X_{q_{2m+1}} + \left[ \ind_{\beta_m,1} + \ind_{\beta_m,3} \right]X_{q_{2m+1}}X_{a}\\
    &+ h_{e_A}X_{a}X_b + h_{e_B}X_bX_{p_0}.
\end{align*}
In words: in the summands in the first four lines, the $XX$ transitions are allowed if and only if the corresponding vertex has its inflow satisfied.
In the last line, the $X_aX_b$ transition is allowed if and only if $e_A$ can be reversed.
Likewise, the $X_bX_{p_0}$ transition is allowed if $e_B$ can be reversed.
If $\Gamma$ is chosen to have an odd number of $\OR$ vertices\footnote{One can always do this without compromising hardness; simply add an $\OR$ vertex to $\Gamma$ and make all of its edges self-loops.\\}, then one can only form an odd-length fundamental cycle $C$ when every PMR term of $H_\gadget$ appears in the fundamental cycle an odd number of times.

We will argue this.
In any closed walk, each qubit must be flipped an even number of times.
Since the gadget's $XX$-terms form a single cycle on the ancillas, the parity vector of gadget-edge usages is either identically zero or identically one.
Because the gadget cycle has odd length, a closed walk has non-trivial holonomy if and only if every gadget edge is used an odd number of times.
The PMR terms of $H_{\text{gadget}}$ commute among one another: their diagonal components act only on the edge qubits $\{b(e)\}_{e \in E}$, which their permutation components (supported on the ancillas) leave untouched.
Consequently, every odd-length fundamental cycle is the same (i.e. equal to $C$).

Notice that, given a bitstring $z \in \{0,1\}^{|E|}$ in connected component $K \subseteq G_H$, then $z$ represents a legal configuration for $\Gamma$ and the transitions $X_{b(e_A)}$ and $X_{b(e_B)}$ are allowed in $K$ if and only if $\bra{z}C\ket{z} \neq 0$.
Also notice that $H_{\gadget} \geq 0$ entry-wise.
So, $H' :=H + H_\text{gadget}$ has VGP if and only if $G_{H'}$ does not contain an odd-length cycle.
To see this, note that in any closed walk every qubit is flipped an even number of times; the $H$-edges (each flipping one $b(e)$ qubit, with negative entries) therefore occur an even number of times, and by \Cref{def:vgp} the holonomy of the walk reduces to $(-1)^{\#\text{gadget edges}}$, the gadget edges being the entry-wise nonnegative ones.
Moreover, the gadget's $XX$ terms link the ancillas $p_0, \ldots, p_\ell, q_0, \ldots, q_{2m+1}, a, b$ into a single ring; since every ancilla must also be flipped an even number of times, the multiset of ring-edge uses in a closed walk has all-even or all-odd multiplicities, and in the all-odd case the number of gadget edges used is congruent modulo $2$ to the ring length $\ell + 2m + 5$, which is odd precisely when the number of $\OR$ vertices is odd.
Hence, a closed walk has non-trivial holonomy exactly when it uses each ring edge an odd number of times, i.e., exactly when it has odd length.

We now proceed with proving that there exists a path $P$ between two legal configurations $x$ and $y$ such that reversal of $e_A$ and $e_B$ are legal at some point along $P$ if and only if $G_{H'}$ does not have an odd-length cycle.
For the forward direction, let $P$ be the path in $G_H$ from $x$ to $y$.
We can form the following odd-length cycle in $G_{H'}$:
\begin{enumerate}
    \item Starting from $\ket{x}$, apply every PMR term from $H_{\text{gadget}}$ to $C$ \emph{except} $h_{e_B}X_bX_{p_0}$.
    \item Next, apply the path $P$.
    \item Now the bit assignments to qubits $\{b(e)\}_{e \in E}$ agree with the configuration $y$. So, we can apply $h_{e_B}X_bX_{p_0}$.
    \item Finish the closed walk by applying the path $P$ in reverse.
\end{enumerate}

For the reverse direction, 
\begin{itemize}
    \item $G_{H'}$ contains an odd-length cycle $\gamma$ $\iff$
    \item $H'$ contains an odd-length fundamental cycle $C$ that admits $\gamma$ into $G_{H'}$ $\iff$
    \item $C$ consists of PMR terms from $H_{\text{gadget}}$, each PMR term occurring as a factor an odd number of times in $C$ $\iff$
    \item There exists a configuration $x$ and a sequence of valid moves reversing $e_A$ and $e_B$.
    \end{itemize}
It suffices to show that such an $x$ must be legal.
If $x$ is illegal, then by~\Cref{lem:edge_locking_lemma}, there is always a vertex $u$ of $\Gamma$ that does not have its minimum inflow constraint satisfied under configuration $x$.
Hence a PMR term in $H_{\text{gadget}}$ evaluates to $0$ on every configuration in the connected component containing $x$ and $y$.
But the only way to have an odd-length cycle in $G_{H'}$ is if there is a fundamental cycle with \emph{every} PMR term of $H_{\text{gadget}}$ as factors.
So $\gamma$ is not a valid walk in $G_{H'}$, contradiction.

Since $\Gamma$ has bounded bandwidth, we may embed the qubits $\{b(e)\}$ in a one-dimensional lattice $\Lambda$ so that $H$ has geometric locality.
Additionally, we require that $e_A$ and $e_B$ lie within $O(1)$ distance of each other in $X$\footnote{Ref. \cite{DeBiasi_2016} does this explicitly in their proof that $\FREENCLREV$ is $\PSPACE$-hard. In Figure 5~\cite{DeBiasi_2016}, the hardness proof chooses $e_A$ and $e_B$ to be incident to vertices that share an edge in a bounded bandwidth constraint graph.}.
Next, embed the ancilla qubits $p_0,\ldots,p_\ell$ and $q_0,\ldots,q_{2m+1}$ in $\Lambda$ so that each ancilla is within $O(1)$ distance of its interacting qubits in $H_{\text{gadget}}$, following the bounded bandwidth ordering of $\Gamma$.
Since $e_B$ is red, there is another blue edge incident to the same vertex as $e_B$---identify $p_0$ with that blue edge when embedding into $\Lambda$ to guarantee the term $h_{e_B}X_bX_{p_0}$ does not compromise geometric locality.
We may do the same for the vertices identified with $q_0$ and $p_\ell$ by the same reasoning.
This renders $H'$ geometrically local and completes the proof.
\qed

\begin{corollary}
    Suppose $H$ is a geometrically local Hamiltonian. Then deciding if $G_H$ is bipartite is $\PSPACE$-complete. $\PSPACE$-hardness is proven for $H \geq 0$ entry-wise and 5-local. \qed
\end{corollary}

\section{Proof of~\Cref{thm:diag_stoq_hardness}}
\label{app:proof_of_nphardness_pmrlocal_vgprecognition_promiseproblem}
Membership in $\NP$ is immediate: the witness is $\theta \in \{0, 2\pi/3\}^n$; conjugation by the single-qubit diagonal $D(\theta)$ preserves every PMR term and its support, so $D(\theta)HD(\theta)^\dagger$ remains PMR-local and its stoquasticity can be checked efficiently (\Cref{prop:efficient_stoq_check_pmr_local}).

We prove $\NP$-hardness by reducing from the $\NP$-complete \emph{monotone 1-in-3-SAT decision problem}~\cite{Garey_1979}.

\begin{definition}
    \label{def:monotone_1-in-3_SAT}
    The \emph{monotone 1-in-3 SAT decision problem} is as follows:
    \begin{itemize}[leftmargin=5.5em]
    \item[\textbf{Input}:] $\psi$ is a Boolean formula in conjunctive normal form with exactly 3 positive literals per clause.
        \item[\textbf{Problem}:] Decide if there exists a Boolean assignment $x \in \{0,1\}^n$ to the $n$ variables of $\psi$ such that exactly one literal evaluates to true per clause of $\psi$. \qed
    \end{itemize}
\end{definition}

Now we prove the rest of~\Cref{thm:diag_stoq_hardness}.

\begin{proof}
    Choose
    \begin{align*}
        \mathcal{D} = \left\{ \bigotimes_{j=1}^n \begin{pmatrix}
            1 & 0\\
            0 & e^{i\theta_j}
        \end{pmatrix} : (\theta_1,\ldots,\theta_n) \in \{0,2\pi/3\}^n \right\}.
    \end{align*}
            
    Fix the $\NP$-complete monotone 1-in-3-SAT decision problem and let $V$ be its polynomial-time verifier which accepts inputs Boolean formula $\psi$ and assignment $x$ if and only if for each clause of $\psi$, exactly one literal is true. Suppose $\psi$ is an arbitrary Boolean formula with 3 literals per clause and let $n$ be the number of variables in $\psi$.

    We will reduce the monotone 1-in-3-SAT problem to determining if there exists a diagonal unitary $D =D(\theta)$ of the form
    \begin{align*}
        D(\theta) = \bigotimes_{j=1}^n \begin{pmatrix}
            1 & 0\\
            0 & e^{i\theta_j}
        \end{pmatrix}, && \theta = (\theta_1,\ldots,\theta_n) \in \{0,2\pi/3\}^n
    \end{align*}
    such that $D(\theta)$ stoquastizes a Hamiltonian $H$.

    Notice that if we define bits
    \begin{align*}
        b_j = \begin{cases}
            0 & \theta_j = 0,\\
            1 & \theta_j = 2\pi/3,
        \end{cases}
    \end{align*}
    then for any computational basis state $\ket{x_1x_2x_3}_{i,j,k}$
    \begin{align*}
        &D\ket{x_1x_2x_3}_{i,j,k} = e^{i(\theta_ix_1+ \theta_jx_2+ \theta_kx_3)}\ket{x_1x_2x_3}_{i,j,k} && \implies\\
        &D\ket{x_1x_2x_3}_{i,j,k} = \omega^{b_1x_1 + b_2x_2 + b_3x_3}\ket{x_1x_2x_3}_{i,j,k}
    \end{align*}
    where we define $\omega := e^{i2\pi/3}$. 
    Define the above exponent $\phi(x) = b_1x_1 + b_2x_2 + b_3x_3$ for every computational basis state $x= \ket{x_1x_2x_3}$.
    
    For every clause $C$ on variables $i,j,k$ in formula $\psi$, define the Hamiltonian
    \begin{align*}
        H_{C} := -\omega \ket{000}\bra{111}_{i,j,k} - \omega^{2}\ket{111}\bra{000}_{i,j,k}.
    \end{align*}
    Then $H_{C}$ is 3-local and only has two non-zero entries, which are $\bra{000}H_{C}\ket{111} = -\omega$ and $\bra{111}H_{C}\ket{000} = -\omega^2$. $H_C$ is also Hermitian. Since $D(\theta)$ is diagonal, for any operator $A$ and any computational basis states $\ket{a},\ket{b}$ we have
    \begin{align*}
        \bra{a}D(\theta)AD(\theta)^\dagger\ket{b} = \omega^{\phi(a) - \phi(b)} \bra{a}A\ket{b}.
    \end{align*}
    After conjugating $H_{C}$,
    \begin{align*}
        q &:= \bra{000}D(\theta)H_{C}D(\theta)^\dagger\ket{111}\\
        &= \omega^{\phi(\ket{000})-\phi(\ket{111})}\left(-\omega\right),
    \end{align*}
    where
    \begin{align*}
        \phi(\ket{000}) -\phi(\ket{111})
        &= -(b_1 + b_2 + b_3) =: -m_C.
    \end{align*}
    Then
    \begin{align*}
        q=-\omega^{1-m_C}.
    \end{align*}
    It is easy to check that $q=-1$ if and only if $m_C=1$. In any other case, $q$ is non-real. Therefore, $H_C$ is stoquastic if and only if exactly one of $b_1,b_2$ and $b_3$ is 1.
    
    Let $\mathcal{C} = \{(i,j,k)\}$ be the set of clauses of formula $\psi$. Define the total Hamiltonian
    \begin{align*}
        H := \sum_{C \in \mathcal{C}} H_C.
    \end{align*}
    Since the sum of stoquastic terms is stoquastic and, for every pair of distinct clauses $C_1$ and $C_2$, we have that $H_{C_1}$ and $H_{C_2}$ do not overlap on the same non-zero matrix entries, $H$ is stoquastic if and only if every $H_C$ is stoquastic.
    Another consequence that follows from this is that $H$ is 3-PMR-local.

    Now, assume $x\in\{0,1\}^n$ is a Boolean assignment and $V(\psi,x)$ is accepted. 
    Define $\theta\in\{0,2\pi/3\}^n$ by setting $b_j:=x_j$ for all variables $j$ and then letting
    \begin{align*}
        \theta_j=\begin{cases}
            0 & b_j=0,\\
            2\pi/3 & b_j=1.
        \end{cases}
    \end{align*}
    Let $D=D(\theta)$. Since $V(\psi,x)$ accepts, every clause $C\in\mathcal{C}$ has exactly one satisfied literal under $x$; equivalently, $m_C=1$ for all $C$. By the calculation above this implies
    \begin{align*}
        \bra{000}DH_CD^\dagger\ket{111} = \bra{111}DH_CD^\dagger\ket{000} = -1
    \end{align*}
    and all other off-diagonal entries of $DH_CD^\dagger$ are zero. Hence each $DH_CD^\dagger$ is stoquastic, and therefore so is their sum
    \begin{align*}
        DHD^\dagger=\sum_{C\in\mathcal{C}} DH_CD^\dagger .
    \end{align*}

    Conversely, suppose there exists $\theta\in\{0,2\pi/3\}^n$ such that $D(\theta)HD(\theta)^\dagger$ is stoquastic. 
    Since distinct clause terms $H_{C_1}$ and $H_{C_2}$ do not overlap on off-diagonal matrix entries, no cancellation of a positive or complex off-diagonal element is possible in the sum. 
    Thus $DHD^\dagger$ is stoquastic only if each $DH_CD^\dagger$ is stoquastic. 
    By the previous argument this holds if and only if $m_C=1$ for every clause $C$. 
    Defining a Boolean assignment $x$ by $x_j:=b_j$ for all $j$, we get that in every clause exactly one literal is true, i.e., $V(\psi,x)$ accepts.
    Therefore determining whether a 3-PMR-local Hamiltonian can be stoquastized by such a local diagonal unitary is $\NP$-hard.
\end{proof}

\end{document}